\documentclass[letterpaper,twocolumn,10pt]{article}
\usepackage{usenix-2020-09}
\usepackage{amssymb}
\usepackage{amsfonts}
\usepackage{mathrsfs}
\usepackage{amsmath} 
\usepackage{amsthm}
\usepackage{mathtools}

\newcommand{\ceil}[1]{ {\lceil#1\rceil}}

\newcommand{\bi}{{\{0,1\}}}

\usepackage{algorithm}
\usepackage[noend]{algpseudocode}
\algrenewcommand\algorithmiccomment[1]{\textcolor{lightgray}{\hfill // #1}}

\usepackage{pgfplots}
\pgfplotsset{compat=1.18}
\usepackage{graphicx}

\usepackage{hyperref}
\hypersetup{
 colorlinks,
 linkcolor={blue!100!black},
 citecolor={blue!100!black},
 urlcolor={blue!80!black} 
}
 
\usepackage{diagbox}

\usepackage[style=ieee,
            backend=biber,
            doi=false,
            ]{biblatex}
 
\usepackage{subcaption}
\usepackage{wrapfig}
\usepackage{pdfpages}

\theoremstyle{definition}
\newtheorem{theorem}{Theorem}

\newtheorem{lemma}{Lemma}

\newtheorem{remark}{Remark}

\usepackage{xcolor}
\usepackage{soul} 

\newcommand{\blue}[1]{#1}

\newcommand{\bbF}{\mathbb{F}}

\newcommand{\bfc}{\mathbf{c}}
\newcommand{\bfd}{\mathbf{d}}

\newcommand{\bff}{\mathbf{f}}

\newcommand{\bfm}{\mathbf{m}}

\newcommand{\bfp}{\mathbf{p}}

\newcommand{\bfr}{\mathbf{r}}
\newcommand{\bfs}{\mathbf{s}}

\newcommand{\bfu}{\mathbf{u}}

\newcommand{\bfw}{\mathbf{w}}

\newcommand{\cC}{\mathcal{C}}

\algdef{SE}[REPEAT]{Repeat}{EndRepeat}[1]{\textbf{repeat}\ #1 \textbf{times:}}{\textbf{end}}
\algtext*{EndRepeat}

\algdef{SE}[DOWHILE]{Do}{doWhile}{\algorithmicdo}[1]{\algorithmicwhile\ #1}%

\renewcommand*{\thefootnote}{\fnsymbol{footnote}}

\begin{document}
\title{Secure Information Embedding in Forensic 3D Fingerprinting}
 
\author{
    {\rm Canran Wang$^*$},
    {\rm Jinwen Wang$^*$},
    {\rm Mi Zhou},
    {\rm Vinh Pham},
    {\rm Senyue Hao},\\
    {\rm Chao Zhou},
    {\rm Ning Zhang},
    {\rm and Netanel Raviv}\\ 
    Washington University in St. Louis
}

\pagestyle{empty}

\newcommand\bsub[1]{\vspace{3pt}\noindent\textbf{#1}}
\newcommand\isub[1]{\vspace{3pt}\noindent\textit{#1}}

\maketitle
 
\renewcommand*{\thefootnote}{}
\footnote{$^*$These authors contributed equally to this work.}

\renewcommand*{\thefootnote}{\arabic{footnote}}
\setcounter{footnote}{0}

\begin{abstract}

    Printer fingerprinting techniques have long played a critical role in forensic applications, including the tracking of counterfeiters and the safeguarding of confidential information. 
The rise of 3D printing technology introduces significant risks to public safety, enabling individuals with internet access and consumer-grade 3D printers to produce untraceable firearms, counterfeit products, and more.
This growing threat calls for a better mechanism to track the production of 3D-printed parts.

Inspired by the success of fingerprinting on traditional 2D printers, we introduce SIDE (\textbf{S}ecure \textbf{I}nformation Embe\textbf{D}ding and \textbf{E}xtraction), a novel fingerprinting framework tailored for 3D printing.
SIDE addresses the adversarial challenges of 3D print forensics by offering both secure information embedding and extraction.
First, through novel coding-theoretic techniques, SIDE is both~\emph{break-resilient} and~\emph{loss-tolerant}, enabling fingerprint recovery even if the adversary breaks the print into fragments and conceals a portion of them.
Second, SIDE further leverages Trusted Execution Environments (TEE) to secure the fingerprint embedding process.

\end{abstract}

\section{Introduction}\label{section:introduction}

3D printing is revolutionizing the consumption and distribution of goods, but also introduces unprecedented security risks that are absent in traditional 2D printing.
With internet access and commercial 3D printers, individuals can fabricate untraceable firearms (\emph{ghost guns}~\cite{Yahoo,mcnulty2012toward}) and other illicit items, with little to no technical expertise.
For example, one such weapon was implicated in the recent killing of Brian Thompson, the CEO of UnitedHealthcare~\cite{gibbons2024when, greenberg2024ghost} (Fig.~\ref{fig:ghostgun}).
To assist authorities and law enforcement in addressing these threats, forensic techniques offer a promising path forward.

\begin{figure}
	\centering
	\begin{subfigure}[b]{0.24\textwidth}
		\includegraphics[height=2.8cm]{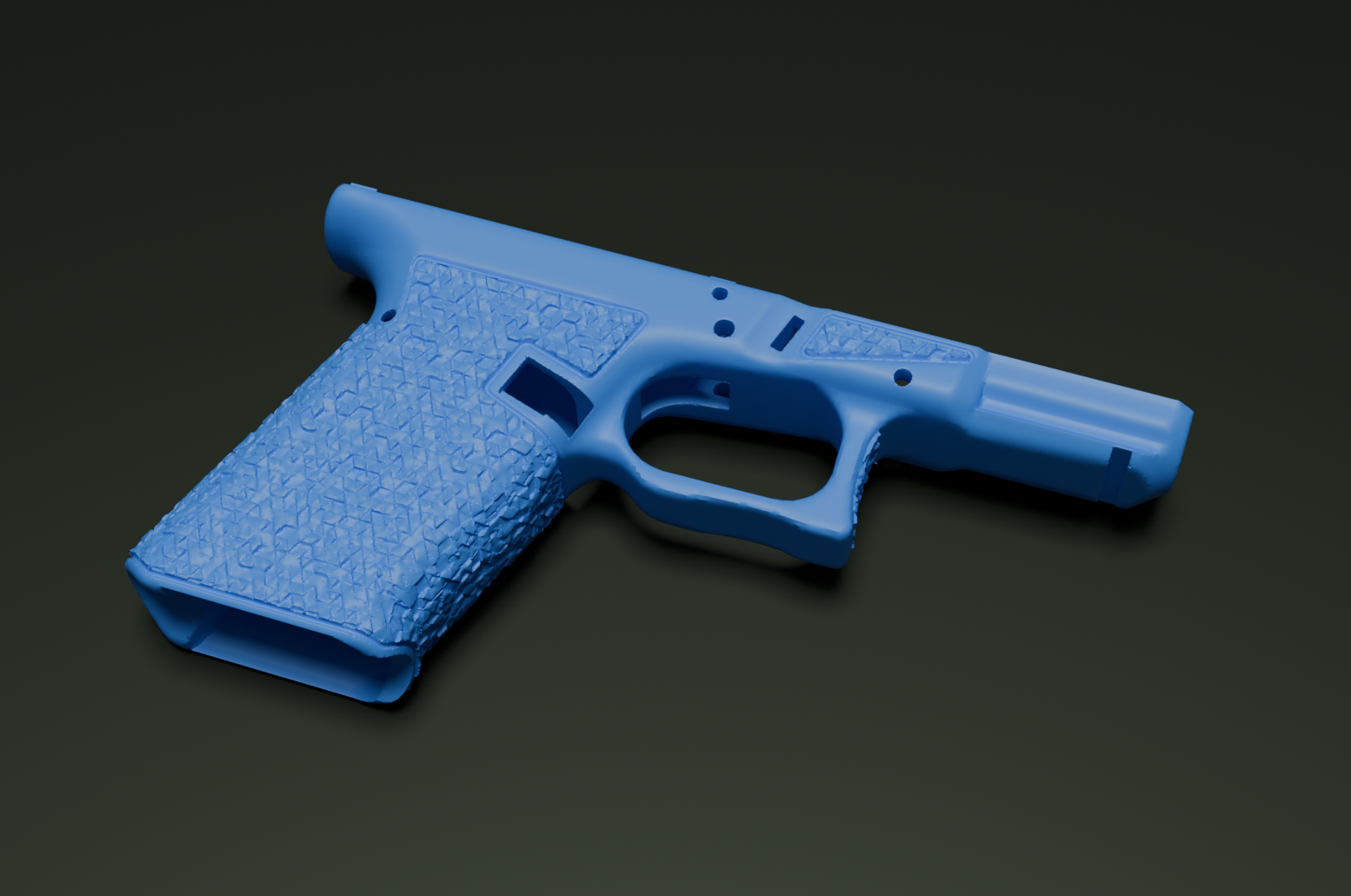}\caption{}\label{fig:g19}\end{subfigure} 
        \hfill
	\begin{subfigure} {0.22\textwidth}
		\includegraphics[height=2.8cm]{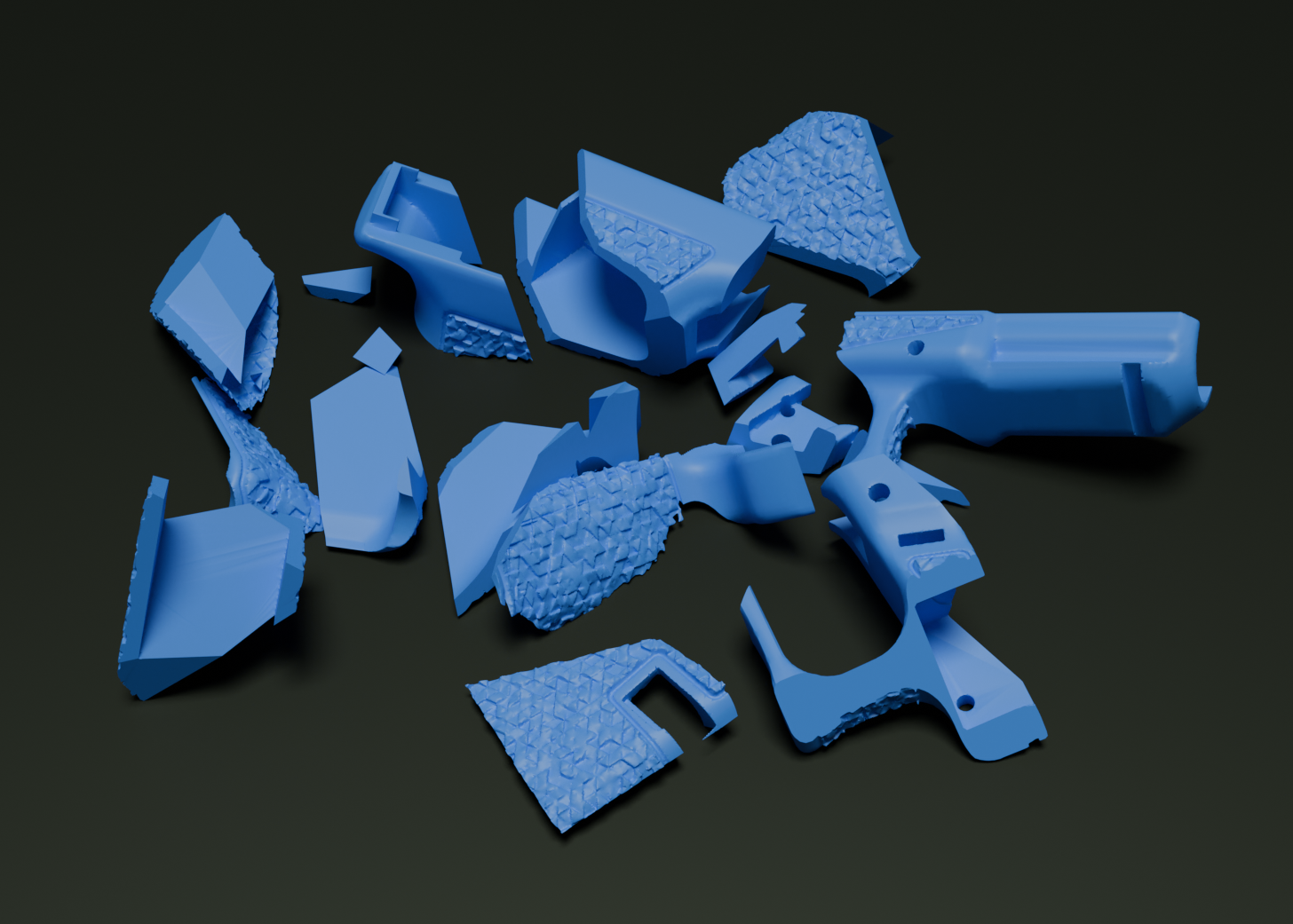}\caption{}\label{fig:g19fractured}
	\end{subfigure}
	\caption{(a) Renderings of a Glock 19 frame design. (b) Fingerprinting fails when the frame is broken, with portions missing.}
 
    \label{fig:ghostgun}
\end{figure}

\bsub{Existing Fingerprinting Solutions: }
Fingerprinting is a widely employed forensic technique that embeds uniquely traceable data into printed documents, e.g. timestamps, geolocations, and printer IDs.
These embedded fingerprints can subsequently be extracted to trace the perpetrator.
In the realm of 2D printer fingerprinting, existing approaches generally fall into two broad categories:~\emph{active} and~\emph{passive} methods.
Active methods involve deliberately placing invisible markers (e.g., a grid of dots~\cite{schoen2005investigating} or traits from modulated laser intensity~\cite{chiang2006extrinsic}).
Passive methods, on the other hand, rely on inherent variations of individual printers, including imperfections or speed fluctuation patterns in specific printer components~\cite{ali2003intrinsic}. 
Extending these concepts to 3D printing, current fingerprinting efforts have explored a variety of techniques, e.g., embedding tags on the object surface by varying layer thickness~\cite{delmotte2019blind} or printing speed~\cite{elsayed2021information}, altering layer material~\cite{salas2022embedding}, inserting cavities~\cite{suzuki2017embedding}, embedding RFID tags~\cite{voris2017three} or QR codes~\cite{wei2018embedding, chen2019embedding}, and inserting acoustic barcodes using surface bumps~\cite{harrison2012acoustic}. 
In addition, some studies focus on intrinsic signatures unique to 3D printers, including characteristics of stepper motors~\cite{li2018printracker} and heat systems~\cite{gao2021thermotag}.
Despite the comprehensiveness of vectors investigated in previous studies, there is little focus on the resiliency of the fingerprint against an active adversary, who may tamper with the embedding software or destroy the fingerprint by breaking the printed tools.

\begin{figure}[!t]
    \centering
\includegraphics[width=0.4\textwidth]{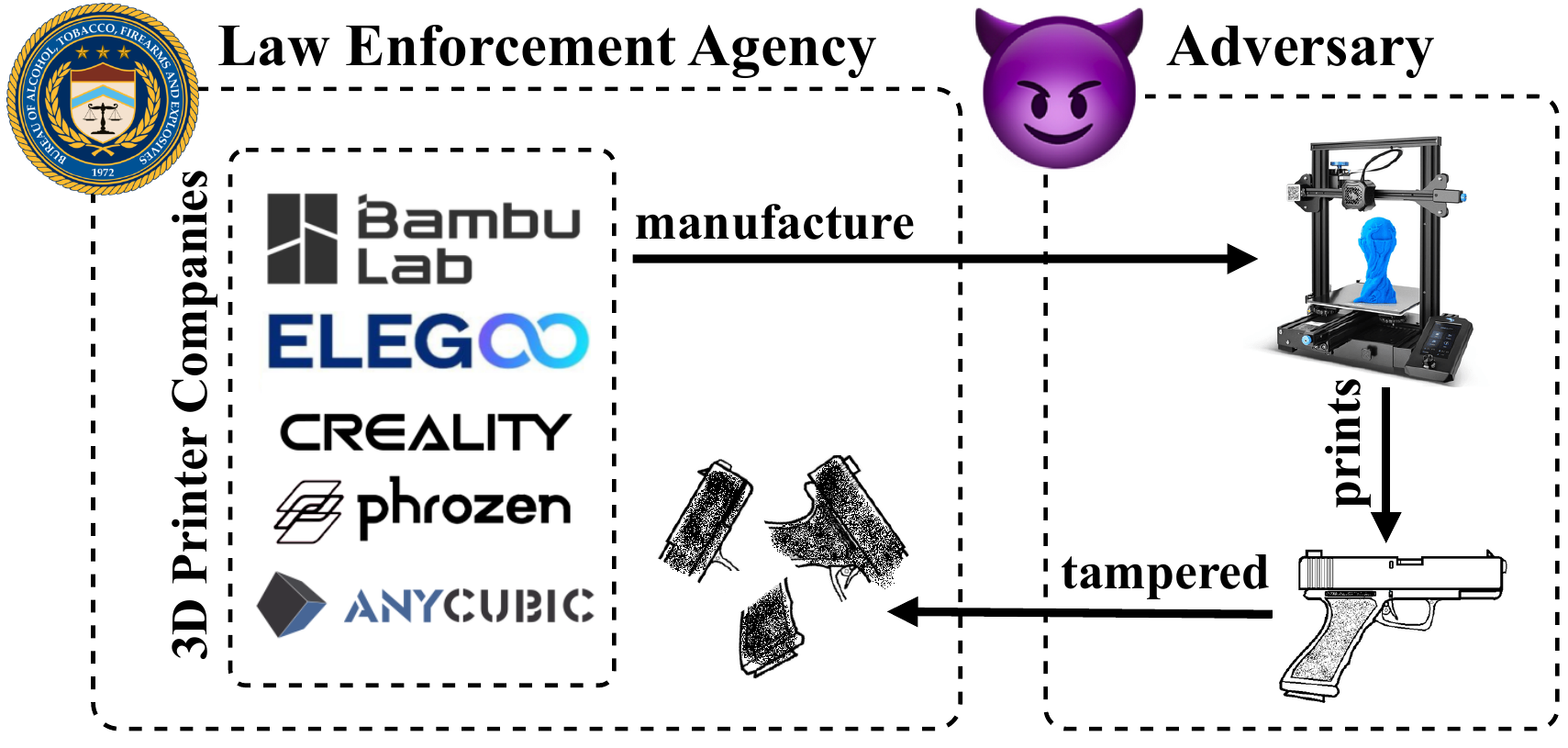}
    \caption{Law Enforcement Agencies Fail Fingerprinting}
    \label{fig:law-vs-adversary}
\end{figure}

\bsub{Resilient 3D Fingerprinting Solution under Adversaries: } 
To protect the integrity of 3D printing in adversarial scenarios, we proposed SIDE, a framework for \textbf{S}ecure \textbf{I}nformation Embe\textbf{D}ding and \textbf{E}xtraction, designed to enhance the resiliency of the fingerprint.
This necessitates cyber-physical protection of the fingerprint in both the physical printed part and the digital process that embeds the fingerprint. In this work, SIDE tackles two key technical challenges. 

\isub{Preventing Information Corruption: }
An adversary who is fully aware of the embedding method can tamper with the embedded fingerprint by breaking the object and hiding some of the fragments, hindering the extraction of the embedded information. 
Existing fingerprinting solutions~\cite{roth2006coding}, which embed raw information into objects, cannot defend against such adaptive adversaries due to their lack of resilience to missing information. 
To address this challenge, we propose 
$\alpha$-\emph{break-resilient codes} ($\alpha$-BRC), a family of coding-theoretic techniques designed for information extraction from 3D prints, with~$\alpha$ serves as a security parameter. 
These codes improve robustness by allowing successful extraction of embedded bits under fragmentation and partial loss, i.e., even when the print is adversarially broken and some fragments are missing up to the specified security threshold.

\isub{Preventing Fingerprint Embedding Tempering: }
An adversary may also attempt to tamper with the processing that embeds the fingerprint into the physical part. For example, if the fingerprint is embedded by the slicing software, then an attacker can attempt to remove the fingerprint by tampering with G-code instructions. As a result, to maximize resiliency against digital attacks, SIDE inserts the fingerprint embedding subsystem into the process immediately before physical printing in the 3D printing pipeline. 
By leveraging a TEE, SIDE ensures the integrity of the embedding process that translates the digital fingerprint to its physical representation, even against a strong attacker who can compromise the OS kernel.

\bsub{Prototype and Evaluations: }
We implemented and tested a prototype of SIDE, on a Creality Ender 3 3D FDM printer controlled by a Raspberry Pi 3B board, with OP-TEE V3.4 as secure world OS, and extend the results to an Elegoo Mars 4 SLA printer.
To evaluate the efficiency and effectiveness SIDE, we measured both the runtime and memory overhead introduced from SIDE, and the fingerprint extraction success rate.
The latter was measured under adversarial scenarios, with a combination of real-world and simulation-based experiments, with the help of a Leica S9D microscope and optical coherence tomography (OCT) devices. 
In summary, we make the following contributions.

\begin{itemize}
    \item We propose a secure 3D fingerprinting mechanism named SIDE that improves resiliency to adversarial operations, including malicious manipulation on both the information embedding procedure and the 3D prints. 
    
    \item We design and implement coding and decoding mechanisms with break-resiliency and loss-tolerance properties to defend against adversaries attempting to tamper with the fingerprinting information on 3D prints. 
    To the best of our knowledge, practical codes with these specific properties have not been studied before. Additionally, we develop a trusted execution environment for the fingerprinting embedding process to protect its integrity against software attacks.

    \item We implemented a prototype of SIDE on the Creality Ender 3 3D printer controlled by a Raspberry Pi 3B board, and evaluated the efficiency and effectiveness via a combination of real-world and simulation-based experiments on 3D objects with varying shapes and sizes.
        
\end{itemize}

\section{Background}\label{preliminary}

\bsub{3D Printing: }
Additive manufacturing, colloquially referred to as 3D printing, has emerged as a revolutionary technology with profound implications across various industries~\cite{yu2023xcheck}.
In contrast to the traditional~\emph{subtractive} manufacturing during which materials are consecutively removed from the workpiece, 3D printing refers to an~\emph{additive} process of creating a physical object and is typically done by laying down many thin layers of material in succession. 

Numerous technologies have been developed for 3D printing.
By and large, they differ by which material is in use and how layers are formed.
A layer can either be formed by using a nozzle that deposits molten thermoplastics while shifting back and forth on a surface, e.g. Fused Deposition Modeling (FDM), by depositing a layer of liquid polymers and curing it by ultraviolet light exposure, e.g., Stereolithography (SLA), or by binding powdered material using high-energy laser beams, e.g., Selective Laser Sintering (SLS).
 
The additive manufacturing process begins by converting a given 3D model into discrete 2D diagrams using a~\emph{slicer} software.
Each diagram represents a planar cross-section of the model along the printing direction at a certain height.
Then, the slicer creates a series of machine commands to instruct the printer about how to produce the corresponding layers of the diagram sequence.

In this research, we focus on printers based on both FDM and SLA technologies for their prevalence in commodity 3D printers.
The commands for FDM printers, called~\emph{G-Codes}\footnote{Note that there exist other languages for controlling 3D printing hardware, including variants of G-codes and propriety ones. 
Nevertheless, we collectively refer to those as G-codes for clarity.}, include the nozzle movement along the~$x$,~$y$, and~$z$ axes, the extrusion of material, and the temperature of nozzle/bed.
The commands for SLA printers may include the movement of print platform and the exposure time of each layer.
Other aspects based on the printer's capabilities may also be specified in these commands.

A typical FDM printer involves four stepper motors, which are actuators that rotate in discrete angular steps of a constant degree.
Three of the motors control the nozzle movement in the Cartesian space, and one is responsible for filament extrusion.
However, a G-code command specifies only the expected action of the printer hardware in a relatively high level, while the low-level implementation is not addressed.
For example, the command 
\begin{equation*}
    \texttt{G1 X98.302 Y96.831 E15.796 F400}
\end{equation*}
merely instructs the printer to move its nozzle from its current position to location~$(x,y)=(98.302,96.831)$ (with its position relative to the~$z$ axis unchanged), and simultaneously extrude~$15.796$ millimeters of molten thermoplastic filament, at a feed rate (speed) of 400 mm/min.
Completing this operation requires a series of~\emph{stepping events}, and each of them defines the exact timing and direction to trigger one of the four stepper motors for a single angular step.

The printer~\emph{firmware} bridges between the G-code and the printer's hardware, translating commands into precise actuator movements that drive the printing process.
The translation process of a firmware is non-trivial, and has a significant impact on print quality.
For instance, a sudden jerk of the printer nozzle may lead to uneven deposition of print material, compromising or even failing a print.
In contrast, a nozzle movement with smooth velocity change is generally preferred.

\bsub{Trusted Execution Environment: }
Trusted Execution Environments (TEEs) have emerged as a reliable solution for isolating sensitive and critical operations from untrusted software stacks in computing systems. 
Utilizing hardware and/or software isolation mechanisms, ARM TrustZone provides a secure execution space, known as \textit{the secure world}, for trusted software stacks. 
User space applications within this environment are commonly referred to as~\emph{Trusted Applications} (TAs). 
The code and data of secure software in the secure world are protected from access or tampering by untrusted software stacks, which operate in the Rich Execution Environment (REE), also known as~\emph{the normal world}.
Benefiting from diverse TEE solutions provided by CPU vendors for different architectures, such as ARM TrustZone, Intel Software Guard Extensions (SGX), and RISC-V Keystone, applications across various domains have been effectively protected. These domains include access control~\cite{krawiecka2018safekeeper,schwarz2020seng,djoko2019nexus}, cloud computing~\cite{baumann2015shielding,schuster2015vc3,poddar2018safebricks,hunt2020telekine,mo2021ppfl,wang2018interface}, and real-time systems~\cite{wang2022rt,pinto2019virtualization,wang2023ari,wang2023secure,wang2023ip}.

\section{Threat Model}\label{section:threatmodel}

We assume that the adaptive attacker is fully aware of SIDE's embedding and extraction schemes and attempts to bypass the fingerprinting mechanism through either the cyber vector or the physical vector.
Physically, we assume the attacker can kinetically break apart the printed object to corrupt the embedded fingerprint information.
Furthermore, the attacker is able to conceal some but not all of the fragments from the forensic investigation.
This is since physically eliminating all traces of evidence often requires significantly greater domain expertise and specialized training. 
Moreover, we do not believe it is possible to forensically link a printed object to a specific printer if all identifiable fragments can be completely concealed, especially given that physical evidence often plays an important role in criminal trials. 
Digitally, we assume the attacker is capable of modifying the files on the file system, and leveraging exploitation tools to escalate privilege in the system, compromising the rich execution environment (e.g., normal world).
However, the attacker cannot forge signatures for secure boot, and the secure software in TEE is free of vulnerability and can be trusted.

We further assume that attackers capable of building their own 3D printing devices, or capable of purchasing untraceable hardware, are out of scope.
This is since these approaches often require significant expertise, incur substantial costs, and deteriorate the quality of the resulting print.
While SIDE certainly has limitations in defending against resourceful attackers with strong expertise in additive manufacturing, it significantly raises the level of sophistication, prior knowledge, and expertise required from the adversary in order to remain undetected after committing the crime. 

We further assume that hardware used by law enforcement during decoding, such as microscopes or computed tomography (CT) devices, is trusted and inaccessible to adversaries.
Additionally, we assume that the adversaries will not intentionally damage printer components like sensors and actuators, as doing so would degrade print quality.
The printer's processor is assumed to support a Trusted Execution Environment (TEE) and to be reliable.
Lastly, side-channel and denial-of-service (DoS) attacks are considered out of scope.

\section{Break-Resilient Codes}\label{section:BRC}

The secure information extraction feature of SIDE is attributed to the~$\alpha$-\emph{break-resilient codes} ($\alpha$-BRC) specifically developed for forensic fingerprinting purposes, in which~$\alpha$ serves as a security parameter.

The~$\alpha$-BRC includes an encoder which take a binary string~$\bfw\in\bi^k$, i.e., the fingerprint, as input, where
\begin{equation}\label{eq:kalm}
    k={(l-\alpha)\cdot m-1},
\end{equation}
with~$l$ and~$m$ being positive integers such that
\begin{equation}\label{eq:kalm2}
    l>\alpha+1~\mbox{and}~m\geq \ceil{2\log l}+2,
\end{equation}
and outputs a codeword~$\bfc\in\bi^n$.
For positive integers~$s$ and~$t$ satisfying  
\begin{equation}\label{eq:st-condition}
    4\cdot t + {{2s}/{(m+\ceil{\log m} +4)}}\leq 4\cdot \alpha,
\end{equation}
the codeword~$\bfc$ is both~$t$-\emph{break-resilient} and~$s$-\emph{loss-tolerant}.
Specifically, this means that even if
\begin{enumerate}
    \item ($t$-break-resiliency) $\bfc$ is broken into~$t+1$ fragments, and
    \item ($s$-loss-tolerance) some fragments are lost, totaling~$s$ bits,
\end{enumerate}  
then the~$\alpha$-BRC decoder can still recover~$\bfw$ from the remaining (unordered) subset of fragments, thereby  secure information extraction is guaranteed, resolving Challenge~2.

\begin{figure}[t]
    \centering
    \begin{subfigure}{0.4\textwidth}
        \centering
	\includegraphics[width=\textwidth]{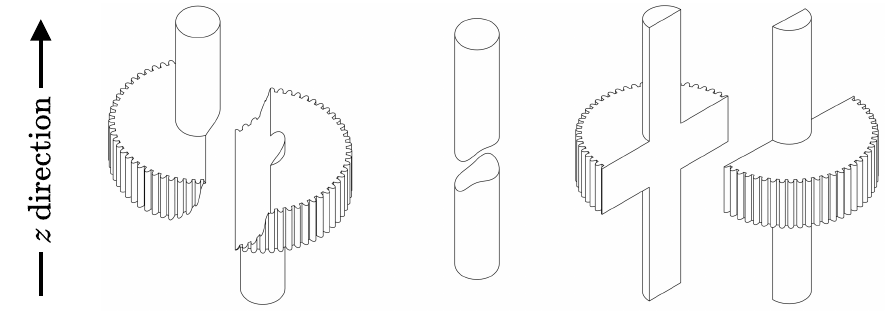}
	\caption{}\label{fig:breaks1} 
    \end{subfigure}
        \\
    \begin{subfigure}{0.43\textwidth}
    \centering
    \includegraphics[width=\textwidth]{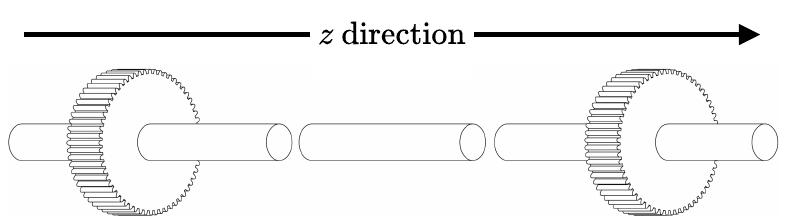}
    \caption{}\label{fig:breaks2}
    \end{subfigure}
    \caption{Fragments of a transmission shaft.
    Breaks in (a) cross multiple layers, and the assembly of fragments may be inferred from their overlapping bits.
    Breaks in (b) are perpendicular to the~$z$ direction, and the correct assembly (i.e., order) cannot be inferred from the fragments themselves.}
\end{figure}
 
\begin{remark}\label{remark:type-of-breaks}
SIDE encodes information into objects using physical elements such as variations in layer thickness.
With this setting, every break falls into one of two categories: those that cross multiple layers (Fig.\ref{fig:breaks1}) and those that do not (Fig.\ref{fig:breaks2}).
Breaks in the former category result in overlapping bits (i.e., bits which are shared between two or more fragments), potentially providing information that enables fragment assembly (see Section~\ref{section:preprocessing}).
Furthermore, if one fragment which contains shared bits is concealed, the shared bits remain accessible through other overlapping fragments.
Breaks in the former category, however, do not create shared bits, and if a fragment in Figure~\ref{fig:breaks2} is concealed, all of its 
embedded bits are omitted from law enforcement.

$\alpha$-BRC are thus designed to handle the worst-case scenario, where~$t$ represents the number of breaks that either produce no overlapping bits or cannot be resolved using overlapping bits, and~$s$ denotes the number of codeword bits entirely absent from all fragments confiscated by law enforcement.
\end{remark}

\subsection{Preliminaries}\label{section:preliminaries}
The following notions from coding theory are employed as basic building blocks for BRC.

\noindent\textbf{Systematic Reed-Solomon (RS) Codes}:
A special type of Reed-Solomon codes, which have been widely employed in communication systems and data storage applications; for an introduction to the topic see~\cite[Ch.~5]{roth2006coding}. 

For integers~$k$ and~$n$ such that~$n>k>0$, a systematic~$[n,k]$ RS code is a set vectors of length~$n$ called codewords, each entry of which is taken from~~$\bbF_q$, a finite field with~$q$ elements.
The first~$k$ entries of each codeword in a systematic RS code carry information in raw form, and the remaining~$n-k$ contain redundant field elements that are used for error correction.
Reed-Solomon codes are \textit{maximum distance separable} (MDS), a property which allows the recovery of a codeword after being corrupted by~$x$ errors (incorrect symbols with unknown locations) and~$y$ erasures (incorrect symbols with known locations), as long as~$n-k\geq 2x+y$.
Further, in this paper we focus on the binary field and its extensions, where~$q=2^z$ for some integer~$z\geq 1$, and any element in~$\bbF_q$ can be represented by a binary string of length~$z$.

\noindent\textbf{Run-length Limited Codes:}
A run-length limited (RLL) code has codewords in which the length of runs of repeated bits is bounded.
We employ the RLL code from~\cite[Algorithm 1]{levy2018mutually} in this paper for its simplicity of implementation.

\noindent\textbf{Mutually Uncorrelated Codes}:
A mutually uncorrelated (MU) code has the property that for every two (possibly identical) codewords, the prefix of one is not identical to the suffix of another.
As such, the codewords of a MU code do not overlap with each other when appearing as substrings of a binary string. 
MU codes have been extensively investigated in the past~\cite{Lev64,Lev70,Gilbert60,Bajic2004,Bajic2014,Chee2013,Bilotta2012,Blackburn2015,Wang2022, levy2018mutually}.
In this paper, we adopt a classic construction of MU code, in which each codeword starts with~$\ceil{\log k}$ zeros followed by a one, where~$k$ is the length of the (binary) information word.
The last bit is fixed to one, and the remaining bits are free from zero runs of length~$\ceil{\log k}+1$.

\noindent\textbf{Distinct Block Codes}.
Inspired by ideas from~\cite[Algorithm 1]{levy2018mutually} and~\cite[Algorithm 1]{elishco2021repeat}, we provide an encoding process that maps an input word to an array of distinct binary strings, and offer the inverse operation.
These procedures serve as an important
component in BRC.
Both algorithms, as well as proofs of their correctness, are given in Appendix~\ref{appendix:distinct-string}.

\begin{algorithm*}[t]\caption{\textsc{Encode} ($\alpha$-BRC Encoding)}\label{alg:Encode}
\begin{algorithmic}[1]
\Statex \textbf{Input:}~{An information word $\bfw\in\bi^k$, where~$k={(l-\alpha)\cdot m-1},$ and~$l,m$ are positive integers s. t.~$m\geq \ceil{2\log l}+2$.}
\Statex \textbf{Output:}~{A codeword $\bfc\in\bi^n$}, where~$n=l\cdot (m+\ceil{\log 
 m}+4) + \alpha\cdot(4m+11)$.
    \State Let~$\bfu\gets\bfp^{(0)}\circ\ldots\circ\bfp^{(\alpha-1)}\circ \bfw$, where~$\bfp^{(i)}$ is the binary representation of~$i\in[0,\alpha-1]$ using~$m$ bits.\label{line:prepending}
    \State Let~$\texttt{dStrings}\gets(\bfu_1,\ldots,\bfu_l)=\textsc{d-encode}(\bfu)$.\label{line:d-encode}
    \State Let~\texttt{next} be a key-value store with keys and values being elements in~$\bi^m$.
    \ForAll{keys~$\bfs$ in $\texttt{next}$ in ascending order}\label{line:defineNextStart}
        \State \textbf{if} there exist~$i\in[0,l-2]$ such that~$\bfs=\bfu_i$
                \textbf{then}~$\texttt{next}[\bfs]\gets\bfu_{i+1}$
        \textbf{else}~$\texttt{next}[\bfs]\gets\bfs$\label{line:defineNextEnd}
    \EndFor
    \State $\bfr_1,\ldots,\bfr_{4\alpha}\gets \textsc{rs-encode}((\texttt{next}[\bfp^{(0)}]\circ 0,\texttt{next}[\bfp^{(1)}]\circ 0\ldots,\texttt{next}[\bfp^{(2^{m-1})}]\circ 0),4\alpha)$,\label{line:rs-encode}
    \For{$i\in [0,\alpha-1]$}
        $\bfd_i \gets \bfr_{4i}\circ \bfr_{4i+1}\circ \bfr_{4i+2}\circ \bfr_{4i+3}$ 
    \EndFor
    \State $\bfc =\textsc{mu}(\bfu_0)\circ\textsc{rll}(\bfd_0) \circ\ldots        \circ\textsc{mu}(\bfu_{a-1})\circ\textsc{rll}(\bfd_{\alpha-1})\circ\textsc{mu}(\bfu_{a})\circ\ldots\circ\textsc{mu}(\bfu_{l-1})$\label{line:codeword}
    \State \textbf{return}~$\bfc$
\end{algorithmic}
\end{algorithm*}

\subsection{Encoding}\label{section:encoding}

The encoding procedure of~$\alpha$-BRC takes a binary string~$\bfw$ as input and outputs a codeword, and it is provided in Algorithm~\ref{alg:Encode}.
At high level,~$\bfw$ will be converted to a sequence of~\emph{distinct} binary strings with the method introduced in Appendix~\ref{appendix:distinct-string}, and their order will be recorded and then protected using a systematic Reed-Solomon code.
Synchronization issues among the symbols of the RS code will be resolved using the RLL and MU techniques mentioned earlier.

Let the information word be~$\bfw\in\bi^k$, and let~$\bfp^{(i)}\in\bi^m$ be the binary representation of integer~$i$.
The encoding of~$\bfw$ begins by prepending~$\bfw$ with $\bfp^{(0)},\bfp^{(1)},\ldots,\bfp^{(\alpha-1)}$, and as shown in line~\ref{line:prepending}, resulting in
\begin{align}\label{eq:u}
    \bfu=\bfp^{(0)}\circ\ldots\circ\bfp^{(\alpha-1)}\circ \bfw\in
    \{0,1\}^{\alpha m+k}\overset{\eqref{eq:kalm}}{=}
    \bi^{l\cdot m-1}.
\end{align}
Then, as shown in line~\ref{line:d-encode}, the resulting string~$\bfu$ is fed into the function~\textsc{d-encode} (Alg.~\ref{alg:djEncoding}, Appendix~\ref{appendix:distinct-string}) and mapped to an array of~$l$~\emph{pairwise-distinct} binary strings of length~$m$, i.e.,
\begin{equation}\label{eq:dStrings}
    \texttt{dStrings}=(\bfu_0,\ldots,\bfu_{l-1}),
\end{equation}
where~$\bfu_i\ne\bfu_j$ for all distinct~$i$ and~$j$ in~$\{0,1,\ldots,l-1\}$.

Note that due to the implementation of~\textsc{d-encode}, the first~$\alpha$ elements of~\texttt{dStrings} remain intact, i.e.,~$\bfu_i=\bfp^{(0)}$ for all~$i<a$, and they are referred as~\emph{markers}.
In the next step, a key-value store~\texttt{next} is defined to represent the~\emph{ordering} of elements in~\texttt{dStrings} as follows (line~\ref{line:defineNextStart}--\ref{line:defineNextEnd}).
For every key~$\bfs\in\bi^m$, the value~$\texttt{next}[\bfs]$ is defined as
\begin{equation*}
    \texttt{next}[\bfs]=\begin{cases}
      \bfu_{i+1}~&\text{if}~\bfs=\bfu_i~\text{for}~i\in[0,l-2],\\
      \bfs~&\text{otherwise.}
    \end{cases}
\end{equation*}

Note that the value~$\texttt{next}[\bfs]$ is well defined, since~$\bfu_i \ne \bfu_j$ for every~$i\ne j$ by the pairwise-distinct property of~\texttt{dStrings}.
It is also worth noting that the mapping from~\texttt{dStrings} to~\texttt{next} is injective; one may recover~\texttt{dStrings} from~\texttt{next} by observing every key~$\bfr$ such that~$\texttt{next}(\bfr)\neq\bfr$, and connecting every two~$\bfr_a,\bfr_b$ of them if~$\texttt{next}(\bfr_a)=\bfr_b$.

We proceed to the treatment of~$\texttt{next}$.
Since the values in~$\texttt{next}$ are binary strings of length~$m$, we append each of them with a~$0$, and hence they can be regarded as symbols in the finite field~$\bbF_{2^{m+1}}$.
They are sorted by their corresponding keys and fed into a systematic Reed-Solomon encoder, which then generates~$4\alpha$ redundancy strings~$\bfr_1,\ldots,\bfr_{4\alpha}\in\bi^{m+1}$ (line~\ref{line:rs-encode}).
Note that such encoding is feasible since the codeword length~$2^m+4\alpha$ is smaller than the number of elements in~$\bbF_{2^{m+1}}$\footnote{
RS codes requires the finite field size to be greater or equal to the length of the codeword, i.e.,~$2^{m+1}\geq 2^m+4\alpha$; this is the case due to~\eqref{eq:kalm2} and the fact that~$\alpha\ge 1$.}.

The codeword~$\bfc$ consists of two parts. 
The first region is called the~\emph{information region}, as it is generated from~$\bfu_\alpha,\ldots,\bfu_{l-1}$, which directly originate from the information word~$\bfw$.
The second region is called the~\emph{redundancy region}.
As the name suggests, it is made from the redundant bits generated from~\texttt{next}.

Define an encoding function~\textsc{mu} which maps~$\bfu_i\in\bi^{m}$ to a codeword~$\textsc{mu}(\bfu_i)\in\bi^{m+\ceil{\log m}+4}$ of a mutually-uncorrelated code~$\cC_\text{MU}$. 
The information region is hereby defined as
$$
\textsc{mu}(\bfu_{\alpha})\circ\ldots\circ\textsc{mu}(\bfu_{l-1})\in\bi^{(l-\alpha)\ldots(m+\ceil{\log m}+4)}.
$$
In addition, for~$i\in[0,\alpha-1]$, define
$$\bfd_i=\bfr_{4i}\circ \bfr_{4\alpha+1}\circ \bfr_{4i+2}\circ \bfr_{4i+3}\in\bi^{4m+4}$$
as the concatenation of four redundancy strings.

Then, let~\textsc{rll} be an encoding function that maps~$\bfd_i$ to a binary sequence~$\textsc{rll}(\bfd_i)\in\bi^{4m +11}$, called~\emph{redundancy packet}, which is free of zero runs longer than~$\ceil{\log m}+1$.
The redundancy region is then defined as
\begin{align*}
    \textsc{mu}(\bfu_0)\circ\textsc{rll}(\bfd_0)\circ
    \textsc{mu}(\bfu_{\alpha-1})&\circ\textsc{rll}(\bfd_{\alpha-1})\\&\in\bi^{\alpha\ldots(5m+\ceil{\log m}+15)}.
\end{align*}
Finally, the codeword~$\bfc$ is the two regions combined (line~\ref{line:codeword}):
\begin{align*}
        \bfc&=\textsc{mu}(\bfu_0)\circ\textsc{rll}(\bfd_0)
        \circ\ldots\circ\textsc{mu}(\bfu_{\alpha-1})\circ\textsc{rll}(\bfd_{\alpha-1})\\
        &\phantom{=}\circ\textsc{mu}(\bfu_{\alpha})\circ\ldots\circ\textsc{mu}(\bfu_{l-1})\in\bi^{l(m+\ceil{\log m}+4)+\alpha(4m+11)}.
\end{align*}

\begin{algorithm*}\caption{\textsc{Decode} ($\alpha$-BRC Decoding)}\label{alg:Decode}
\begin{algorithmic}[1]
\Statex \textbf{Input:}~{A multiset $\texttt{FRAGMENTS}$ of unordered and partially missing fragments of a codeword~$\bfc\in\cC$.} 
\Statex \textbf{Output:}~{The information word~$\bfw$ such that Algorithm~\ref{alg:Encode} with input~$\bfw$ yields $\bfc$.}
\State Let~$\bfr_i\gets\texttt{erasure}$ for all~$i\in[0,4\alpha-1]$.\label{line:defaultErasure}
\ForAll{codeword~$\textsc{mu}(\bfu_i)\in\cC_\text{MU}$ in the fragments in $\texttt{FRAGMENTS}$ and~$\bfu_i=\bfp^{(i)}$ for some integer~$i<\alpha$}\label{line:look4redundancy}
    \If{the number of bits after~$\bfm_i$ is less than~$4m+11$} continue to next MU codeword.\EndIf
    \State Let~$\bfm_i$ be the~$4m+11$ bits after~$\textsc{mu}(\bfu_i)$, and~$\bfd_i\gets \textsc{de-rll}(\bfm_i)$.\label{line:extractdi}
    \State $\bfr_{4i},\bfr_{4i+1},\bfr_{4i+2},\bfr_{4i+3}\gets\bfd_i[0,m],\bfd_i[m+1,2m+1],\bfd_i[2m+2,3m+2],\bfd_i[3m+3,4m+3]$\label{line:extractai}
\EndFor

\State Let~\texttt{approxNext} be a key-value store such that~$\texttt{approxNext}[\bfs]=\bfs$ for all~$\bfs\in\bi^m$.\label{line:approxNext}
\ForAll{fragment~$\bff=\bff_\text{start}\circ\textsc{mu}(\bfu_{u})\circ\ldots\circ\textsc{mu}(\bfu_{u+v})\circ\bff_\text{end}\in\texttt{FRAGMENTS}$ where~$u\ge \alpha$}
    \ForAll{$c\in[0,v-1]$}~$\texttt{approxNext}[\bfu_{u+c}]=\bfu_{u+c+1}$.\label{line:updateApproxNext}
    \EndFor
\EndFor
\State $\texttt{next}\gets\textsc{rs-decode}(\texttt{approxNext},\bfr_1,\ldots,\bfr_{4\alpha})$ \label{line:decodeRS}
\State Let~$\texttt{dStrings}=(\bfu_1,\ldots,\bfu_l)$ be an array of such that~$\texttt{next}(\bfu_i)=\bfu_{i+1}$.\label{line:obtaindstring}
\State $\bfu\gets\textsc{d-decode}(\texttt{dStrings})$
\State \Return $\bfu[\alpha\cdot m:]$\label{line:decodeReturns}
\end{algorithmic}
\end{algorithm*}

\subsection{Preprocessing}\label{section:preprocessing}
    Recall that breaks may cross multiple layers, resulting in overlaps between bit strings extracted from confiscated fragments.
    If a MU codeword is found in two bit strings, the two strings can be merged into one due to the uniqueness of MU codewords (i.e., a MU codeword appears at most once in the BRC codeword).
    This uniqueness arises from the pairwise distinct property of the elements in \texttt{dStrings}, as defined in~\eqref{eq:dStrings}.
     
    Hence, prior to BRC decoding, the bit strings from confiscated fragments undergo a~\emph{preprocessing} stage in which strings that share a MU codeword are merged.

\subsection{Decoding}\label{section:decoding}

Algorithm~\ref{alg:Decode} provides a procedure for extracting the information word~$\bfw$ from the~\emph{unordered} and~\emph{partially missing} fragments of the respective codeword~$\bfc$.
The crux of this procedure is to reconstruct the key-value store~\texttt{next} defined previously, and recover the information word~$\bfw$ from it.

Specifically, the decoding algorithm creates a key-value store~\texttt{approxNext}, which is slightly different from~\texttt{next}, using the information which appears in the confiscated fragments.
Alongside the correctly identified redundancy strings, \texttt{approxNext} goes through a Reed-Solomon decoding process and is corrected to~\texttt{next}.
Having the correct~\texttt{next} in hand, the correct~\texttt{dStrings}~\eqref{eq:dStrings} can be found since the mapping from the latter to the former is injective.
Then, $\texttt{dStrings}$ is fed into \textsc{d-decode} (Alg.~\ref{alg:djDecoding}, Appendix~\ref{appendix:distinct-string}), which is the inverse process of~\textsc{d-encode} (Alg.~\ref{alg:djEncoding}) to produce~$\bfu$~\eqref{eq:u}, whose suffix is the information word~$\bfw$.

In more detail, the decoding starts by distinguishing and decoding the discernible codewords of~$\cC_\textsc{mu}$ from the fragments.
Let~$\textsc{mu}(\bfu_i)$ be a discernible codeword in~$\cC_\text{MU}$ which fully resides within one fragment, where~$\bfu_i$ is its respective decoding. 
If~$\bfu_i=\bfp^{(i)}$ for some integer~$i<a$, it means that~$\bfu_i$ is a marker, and hence the~$(4m+11)$ bits after it consist of a redundancy packet (line~\ref{line:look4redundancy}).
This redundancy packet, if residing in the fragment, is passed to an RLL-decoder which yields four redundancy strings~$\bfr_{4i},\bfr_{4i+1},\bfr_{4i+2},\bfr_{4i+3}$ (line~\ref{line:extractdi}--\ref{line:extractai}).

The other discernible codewords of $\cC_\textsc{mu}$, i.e., those encoded from non-markers ($\bfu_i$'s for~$i\geq a$), are used to construct a key-value store \texttt{approxNext}.
Initially,~$\texttt{approxNext}(\bfs)=\bfs$ for every key~$\bfs\in\bi^m$ (line~\ref{line:approxNext}).
For each fragment~$\bff$, let~
$$\bff=\bff_\text{start}\circ\textsc{mu}(\bfu_{u})\circ\ldots\circ\textsc{mu}(\bfu_{u+v})\circ\bff_\text{end},$$
where~$\bff_\text{start}$ and~$\bff_\text{end}$ are the (possibly empty) prefix and suffix of~$\bff$ with no discernible codeword from~$\cC_\text{MU}$ that is encoded from a non-marker.
In line~\ref{line:updateApproxNext}, the decoder updates
$$\texttt{approxNext}[\bfu_{u+c}]=\bfu_{u+c+1},$$
for every~$c\in[0,v-1]$.
The above process stops once no more codewords in~$\cC_\text{MU}$ can be found.

The decoding algorithm proceeds to correct the constructed key-value store~\texttt{approxNext} to~\texttt{next}, i.e., the key-value store generated in Algorithm~\ref{alg:Encode} from~$\bfw$, using the collected redundancy strings and a standard Reed-Solomon decoder (line~\ref{line:decodeRS}).

Next, the array~\texttt{dStrings} is obtained from \texttt{next} (line~\ref{line:obtaindstring}), and the function~\textsc{d-decode} (Alg.~\ref{alg:djDecoding}, Appendix~\ref{appendix:distinct-string}) is employed to recover~$\bfu$.
Recall that~$\bfu=\bfp^{(0)}\circ\ldots\circ\bfp^{(a-1)}\circ \bfw$ in~\eqref{eq:u}, and hence the decoding procedure concludes by returning the~$k$ rightmost bits of~$\bfu$ (line~\ref{line:decodeReturns}).

Together, correct decoding is guaranteed by the following theorem, whose proof is provided in Appendix~\ref{appendix:proof-of-theorem}.
\begin{theorem}\label{theorem:brc}
Line~\ref{line:decodeReturns} of Algorithm~\ref{alg:Decode} returns correct information word~$\bfw$ if~$    4\cdot t + {{2s}/{(m+\ceil{\log m} +4)}}\leq 4\cdot \alpha$.
\end{theorem}

\subsection{Trusted Information Embedding}

Figure~\ref{fig:tee-printing-process} shows the system design of trusted fingerprint embedding procedures.
To prevent attackers from tampering the fingerprint embedding procedures, SIDE executes the 3D fingerprint embedding procedure and dependencies in TEE, including fingerprint information encoding (\textsf{codec}), object model slicing (\textsf{layer-gen}), toolpath generation (\textsf{toolpath-gen}), and 3D printer driver (\textsf{firmware}).
TEEs are constrained by hardware resource limitations, notably their limited memory for the secure domain.
Consequently, attempting to execute the entire fingerprint embedding procedure na\"{i}vely within the TEE risks print failures due to insufficient secure heap space, which is inadequate to handle the size of 3D model files and the substantial intermediate data generated during the slicing process.
To mitigate this limitation, we propose a~\emph{progressive slicing} strategy, which reduces peak memory usage for procedures with the highest heap memory demand.
This approach leverages the inherent layer-by-layer nature of the 3D printing process.
Instead of pre-slicing the entire 3D model and generating a G-code file prior to printing, slicing is performed dynamically in an on-demand manner during the printing process.

\section{SIDE Implementation}\label{section:implementation}
This section details the implementation of SIDE, including bit embedding, extraction, and TEE integration. 

\subsection{Bit Embedding Method}\label{section:bit-embedding}
\begin{figure}[t]
    \centering
    \includegraphics[width=0.47\textwidth]{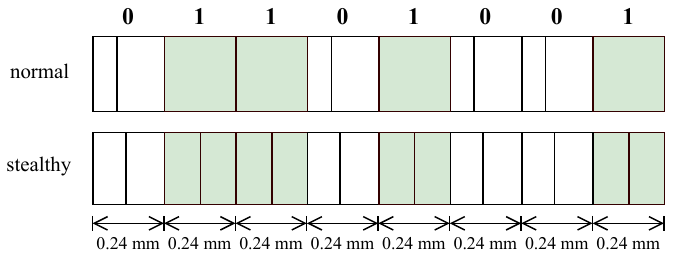}
    \caption{Demonstration of embedding~$01101001$ with parameters~$x=0.08$ and~$(y,\epsilon)=(0.12,0.02)$.
    In the normal settings, each~$0$ is represented by two layers of~$0.08$ and~$0.16$ millimeters, and each~$1$ is represented by one layers of~$0.24$ millimeters.
    In the stealthy settings, each~$0$ is represented by two layers of~$0.10$ and~$0.14$ millimeters, and each~$1$ is represented by two layers of~$0.12$ millimeters.
    In either case, the length required for embedding one bit is~$0.24$ mm.}
    \label{fig:settings}\label{fig:varyLayerWidth}
\end{figure}

We address bit embedding at two levels: normal and stealthy.
The normal approach focuses solely on the readability of the embedded bits, while the stealthy approach imposes an additional requirement of indiscernibility.

\noindent\textbf{Normal Embedding}.
In the proposed normal method, there are three layer thicknesses: $x$, $2x$, and $3x$, where $x$ is a base thickness.  
A~$0$ bit is represented by two consecutive layers of thickness $x$ and $2x$, respectively, and a~$1$ bit is represented by a single layer of thickness $3x$.
This method has three key advantages.
First, it improves the readability as the substantial difference between the layers minimizes confusion during bit reading.
Second, it provides consistent embedding density, defined as the number of bits embedded per unit distance.
Since both $0$ and $1$ are represented using the same total thickness~$3x$, the embedding density remains consistent for both~$0$ and~$1$.
This ensures that the required object height depends solely on the length of the codeword, rather than its content.
Third, it includes directional information, as the arrangement of layers for~$0$'s implies the direction of the codeword.

\noindent\textbf{Stealthy Embedding}.
The stealthy embedding method prioritizes indiscernibility by minimizing differences in layer thickness between layers.
In this method, every bit is represented by two layers.
Specifically, a $1$ is represented by two layers of~$y$ millimeters, and a~$0$ is represented by two layers of~$y-\epsilon$ and~$y+\epsilon$ millimeters, respectively.
The~$\epsilon$ shall be imperceptible to naked eyes but still discernible to bit extraction equipments.
This approach offers the added benefit of stealthiness, making it harder for the adversary to distinguish the embedded bits.
Yet, it demands higher resolution for accurate bit extraction.

An illustration of both embedding methods is given in Figure~\ref{fig:varyLayerWidth}.
In reality, the~$x,y$ and~$\epsilon$ are determined to accommodate the resolutions of the printer and the bit extraction equipment.

\begin{figure}
	\centering
        \includegraphics[width=0.43\textwidth]{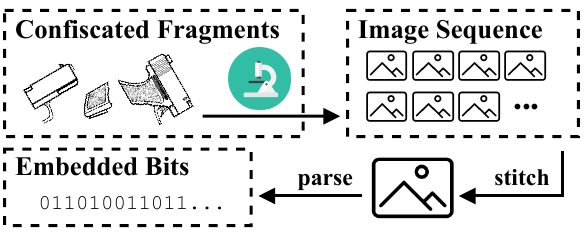}
    \caption{Procedure of bit extraction.}
    \label{figure:bit-extraction}
\end{figure}

\subsection{Bit Extraction Method}\label{section:bit-extraction}
The extraction of bits is performed by inspecting the thicknesses of the layers in the fragments using specialized equipment and determining the corresponding bits.
In this section, we describe an optical extraction method that employs a microscope to inspect the ridges on the surface of the print\footnote{If surface examination is not feasible (e.g., due to post-processing done on the surface), a computed-tomography based method can be utilized. However, it is beyond the scope of this paper}, with a graphical illustration given in Figure~\ref{figure:bit-extraction}.

Microscopes typically have a limited field of view, making it impossible to inspect an entire fragment at once.
To overcome this limitation, we mount the fragment on a motorized rail slider.
During the extraction process, the fragment slides over the microscope's field of view, while the microscope takes a series of pictures.
The pictures are taken so that every two consecutive pictures overlap, which allows us to fuse them together and obtain a picture of the entire fragment.
 
To automatically read bits from the stitched images, we developed \texttt{bit-parser}, a program capable of parsing bits in both normal and stealthy embedding settings.
In the case of normal embedding, \texttt{bit-parser} begins by identifying the layers representing~$1$'s; they are characterized by a single, thickest layer of~$3x$ millimeters, making them easily distinguishable from others.
The program then counts the layers between these thick layers, with each pair of consecutive layers corresponding to a~$0$.

For stealthy embedding, \texttt{bit-parser} examines every two consecutive layers.
If the first layer is thinner than the second by a specific threshold, the pair corresponds to a~$0$.
The program then counts the layers between the layer pairs representing~$1$'s, with each pair of consecutive layers corresponding to a~$1$.
This concludes the process of bit extraction.

\subsection{TEE Protected Embedding}\label{section:implementation-embedding}

\begin{figure}[t]
    \centering
    \includegraphics[width=0.43\textwidth]{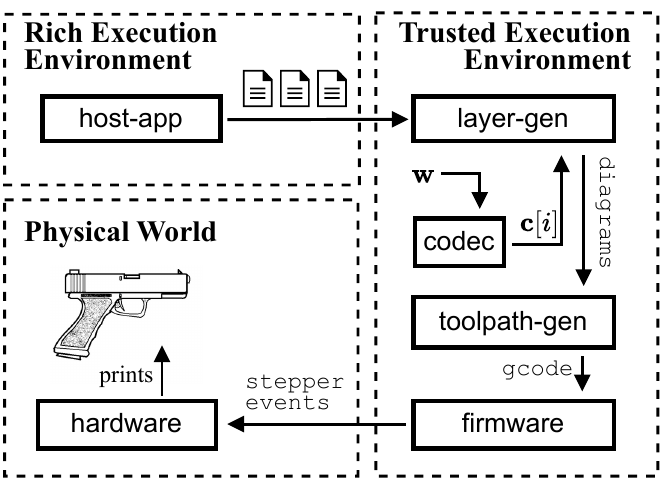}
    \caption{TEE Protected Embedding Procedure}
    \label{fig:tee-printing-process}
\end{figure}

Following the blueprint in Figure~\ref{fig:tee-printing-process}, SIDE involves a~\textsf{host-app} that runs in the normal world to serve as the frontend interface for printer users, and its backend, which consists of four functional modules (\textsf{codec},~\textsf{layer-gen},~\textsf{toolpath-gen}, and~\textsf{firmware}), that fully reside in the TEE.
The host bridges the normal world with the trusted world using C Foreign Function Interface (CFFI).
Upon receiving a 3D model from the user, it first cuts it into consecutive segments in the~$z$ direction, and passes them to SIDE backend one after the other.
The thickness of segments equals to the distance required to represent a bit determined by the bit embedding method.

The actual slicing, as well as the bit embedding, is executed in the backend.
The~\textsf{codec} module is designed to perform encoding, mapping an information word~$\bfw$ to a break-resilient codeword~$\bfc$; the details of the encoding process were given in Section~\ref{section:BRC} and implemented in Algorithm~\ref{alg:Encode}.
Recall that every bit~$\bfc[i]$ instructs the slicing of the corresponding segment.
If~$\bfc[i]=1$, then the segment is sliced to a layer of~$3x$ mm with normal embedding, or~$2y$ mm with stealthy embedding.
Otherwise, it is sliced to two layers of~$x$ mm and~$2x$ mm with normal embedding, respectively, or~$y-\epsilon$ mm and~$y+\epsilon$ mm with stealthy embedding
(see Section~\ref{section:implementation-embedding} for details).
The slicing is performed by the~\textsf{layer-gen} module.
For each layer, it generates the cross-sectional~\texttt{diagram}, and feed them to~\textsf{toolpath-gen} along with their corresponding heights, i.e., their distances to the printer bed.
With these inputs from~\textsf{layer-gen}, the~\textsf{toolpath-gen} generates nozzle toolpath (represented by G-code) used to manufacture these layers.
Both the~\textsf{layer-gen} module and the~\textsf{toolpath-gen} module are developed on top of t43~\cite{t43}.

Finally, the~\textsf{firmware} performs the parsing of G-code generated from~\textsf{toolpath-gen}.
It is a collection of core functionalities provided by Klipper, 
including the computation of precise nozzle movement and the generation of stepper events.
The stepper events are then converted to signals passed to the printer hardware.
This concludes the handling of $\bfc[i]$.
Upon finishing the slicing of a segment, the~\textsf{host-app} module is triggered to feed in the next layer segment, and the printing process is concluded after handling all segments.
Since the entire printing process is hidden in the trusted world and no intermediate data (e.g., a G-code file) is exposed to the user, the adversary is unable to strip off the embedded bits.

In prototyping our design, we employ a Creality Ender~3 3D printer, and a Raspberry Pi 3B board (with OP-TEE V3.4 support enabled in its ARMv8-A architecture, and Raspbian Linux 4.14.98-v7 installed for the normal world) to serve as the control board.
Our development is heavily based on the Klipper open-source project~\cite{Klipper} and t43~\cite{t43}.
The former is a 3D printer firmware known for offering high precision stepper movements offering support to printers with multiple micro-controllers, and is suitable for running on low-cost devices such as Raspberry Pi.
The latter is an open-source slicer program with basic functionalities and is purely written in C, making it suitable for trusted environments with limited language support.

\section{Evaluation}\label{section:evaluation}

This section presents a comprehensive evaluation of SIDE, focusing on fingerprint recovery, its impact on the printing process, and print quality.
Specifically, we provide: (1) Experiments and simulations assessing fingerprint recovery from broken prints.
(2) Analysis of the BRC code rate and the minimum object dimensions required for successful fingerprinting.
(3) Analysis the practicality of stealthy embedding.
(4) Print quality comparisons across normal embedding, stealthy embedding, and no embedding.
(5) Assessment of printer imperfections through analysis and experiments.
(6) Evaluation of TEE integration overhead and its effects on the printing process and print quality.

\subsection{Fingerprint Recovery}

We conducted experiments and simulations to validate fingerprint recoverability.
The experiments replicate the forensic information flow, including BRC encoding, bit embedding, fragmentation, bit extraction, and BRC decoding, demonstrating the core functionality of SIDE.
In simulations, we assess fingerprint recovery success under varying conditions, such as model size and shape, extent of fragmentation, and fragment loss.
The results confirm the robustness and practicality of SIDE for forensic fingerprinting applications.

\begin{figure}[t]
	\centering
	\begin{subfigure}{0.43\textwidth}
		\includegraphics[width=\textwidth]{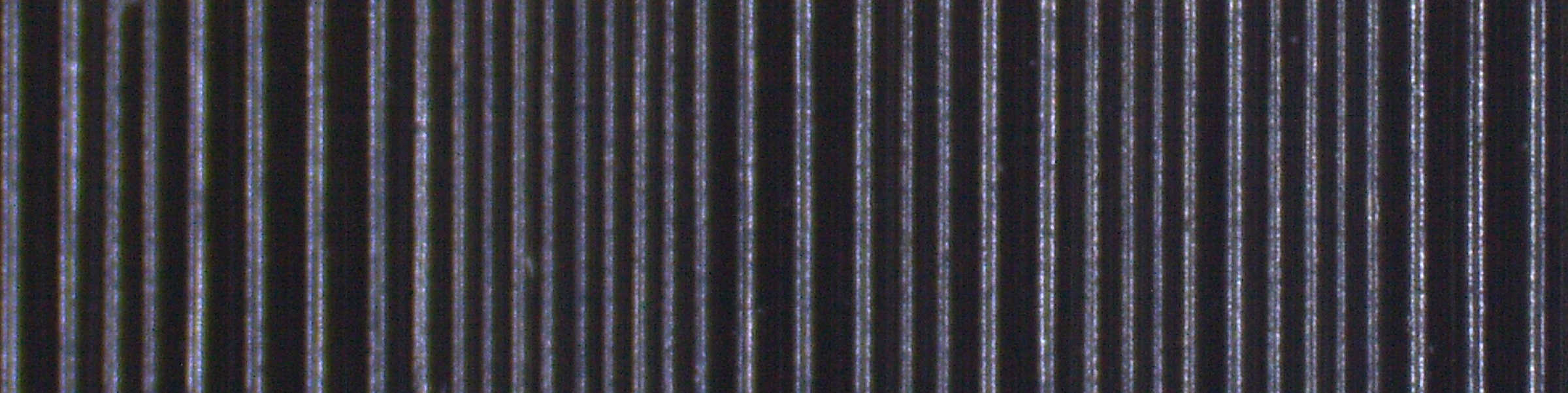}
		\caption{}\label{fig:reading-fdm}
	\end{subfigure}
        \centering
	\begin{subfigure}{0.43\textwidth}
		\includegraphics[width=\textwidth]{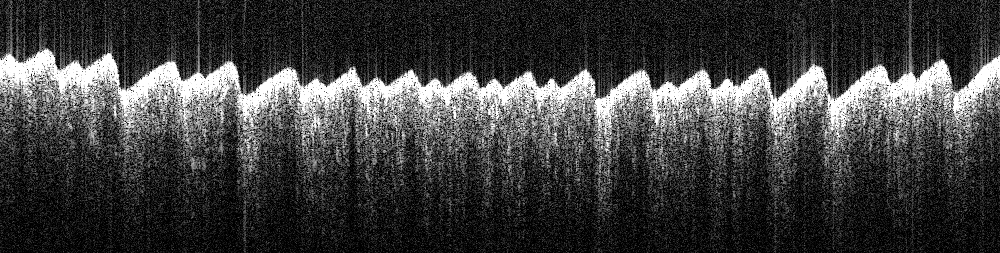}
		\caption{}\label{fig:reading-sla}
	\end{subfigure}
 	   \caption{The readings of (a) an FDM fragment using the Leica S9D microscope, in which every bright line is the center of a layer, and (b) an SLA fragment using OCT equipment, which shows the cross-section of fragment surface, in which a ridge represents a layer. Both fragments are printed with the normal embedding method.}
\end{figure}

\bsub{Real World Scenarios: }
We conduct experiments to verify the recoverability of fingerprints, providing a proof of concept for SIDE.
The experiments mimic the information flow in a forensic scenario, involving BRC encoding, bit embedding, fragmentation, bit extraction, and BRC decoding.
Specifically, we prepared a fingerprint of~$120$ bits, and encoded it into~$1$-,~$2$-, and~$3$-BRC codewords of~$281$,~$353$, and~$425$ bits, respectively.

The experiments were carried out using both FDM and SLA printers, employing the bit embedding method introduced in Section~\ref{section:bit-embedding}.
For the Creality Ender~3 FDM printer, we set~$x = 0.08$, while for the Elegoo Mars 4 SLA resin printer,~$x=0.04$.
These values were selected to balance information density (i.e., the number of bits embedded per unit length) with the resolution capabilities of the respective printers.

For each printer, we printed a cuboid of width~$6$mm and length~$20$mm, while the height is determined by the embedding method and codeword length.
The printed cuboids were then manually broken apart to the maximum allowance of fragmentation by the embedded BRC codeword.

The fragments from FDM printer are examined using the methods described in Section~\ref{section:bit-extraction} with the help of a Leica S9D microscope.
The fragments from SLA printer, however, are examined with an Optical Coherence Tomography (OCT) device, since their layer thickness are beyond the resolution of optical microscopes.
OCT is based on low-coherence interferometry to capture depth-resolved images with micro-level resolution in a non-invasive manner.
It provides fast 3D imaging and quantitative, layer-by-layer analysis.
Although widely used in biomedical and clinical diagnostics, OCT has also been adopted in non-biomedical fields such as industrial inspection~\cite{he2023robotic}, art conservation~\cite{liang2005face}, and geology~\cite{campello2014micro}.
To extract the embedding information from 3D prints, we built our customized~\emph{spectral domain} OCT (SD-OCT) system using visible light with an axial resolution of 1.9 $\mu$m.
 
Finally, the extracted bits were fed to the BRC decoder (Alg.~\ref{alg:Decode}).
In all cases, we achieved a 100\% success rate, i.e., the information word perfectly matched the output of the decoding algorithm.

\begin{figure}[t]
	\centering
	\begin{subfigure}{0.23\textwidth}
		\includegraphics[height=2.3cm]{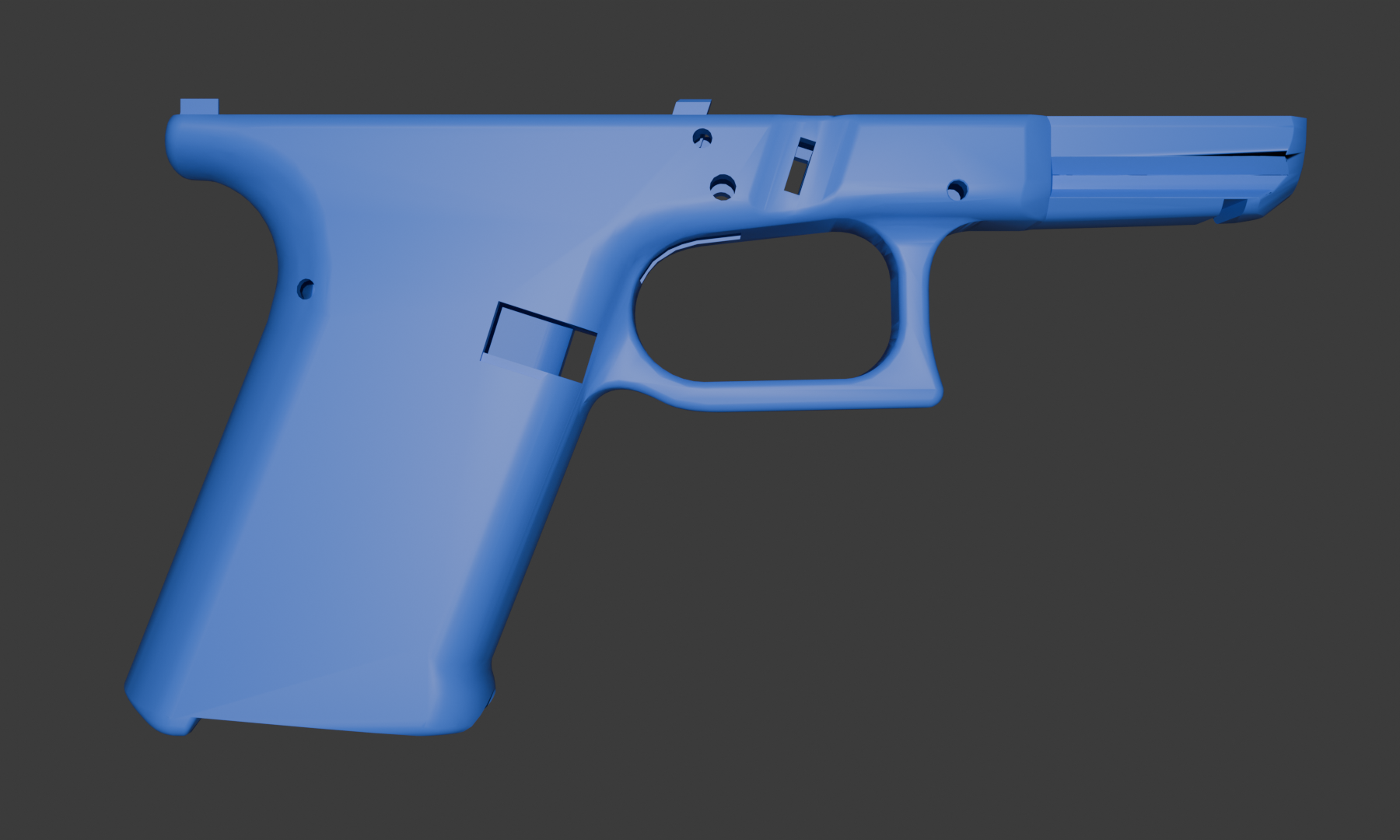}
		\caption{$\beta=0$}
	\end{subfigure}
        \hfill
	\begin{subfigure}{0.23\textwidth}
		\includegraphics[height=2.3cm]{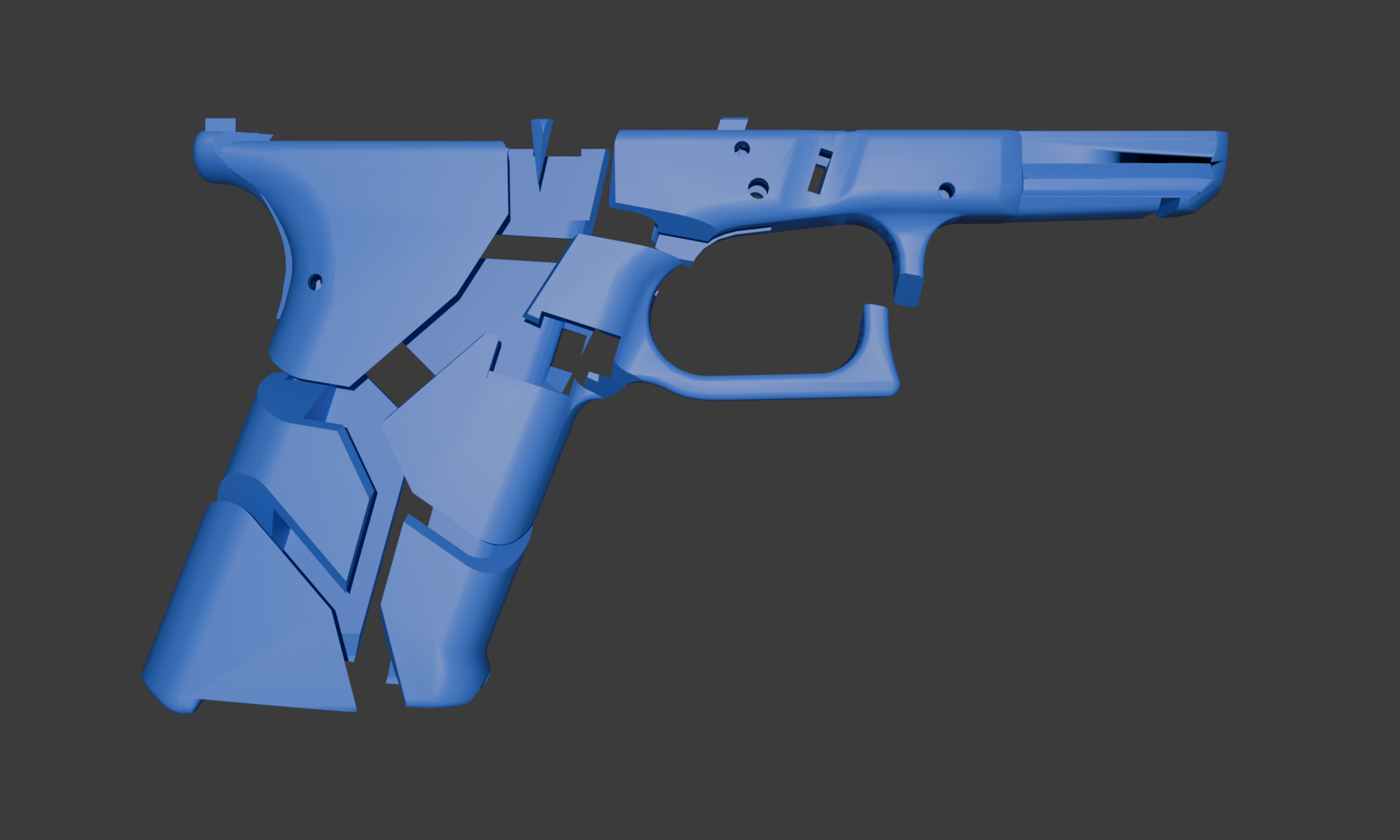}
		\caption{$\beta=10$.}
	\end{subfigure}
    \centering
	\begin{subfigure}{0.23\textwidth}
		\includegraphics[height=2.3cm]{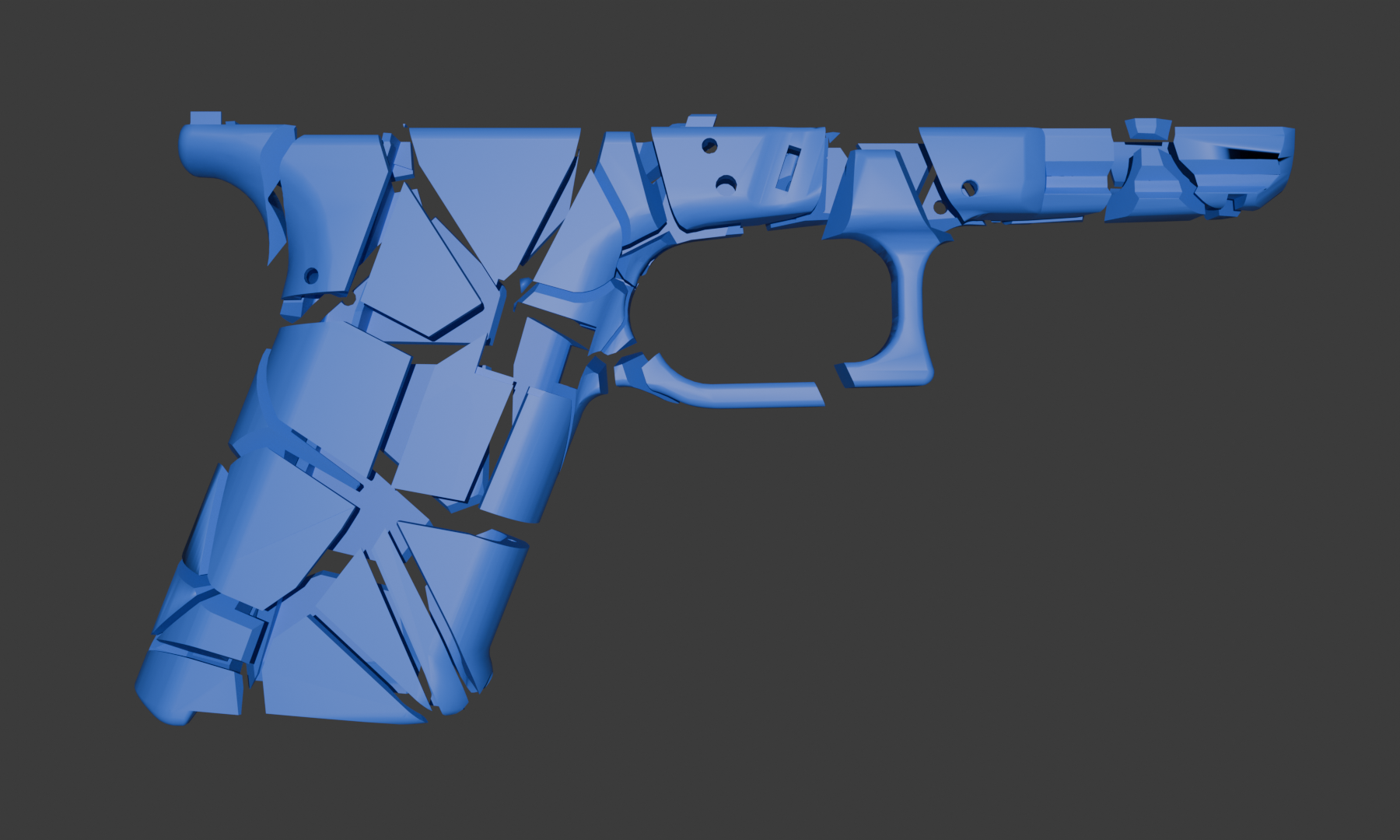}
		\caption{$\beta=50$.}
	\end{subfigure}
        \hfill
	\begin{subfigure}{0.23\textwidth}
		\includegraphics[height=2.3cm]{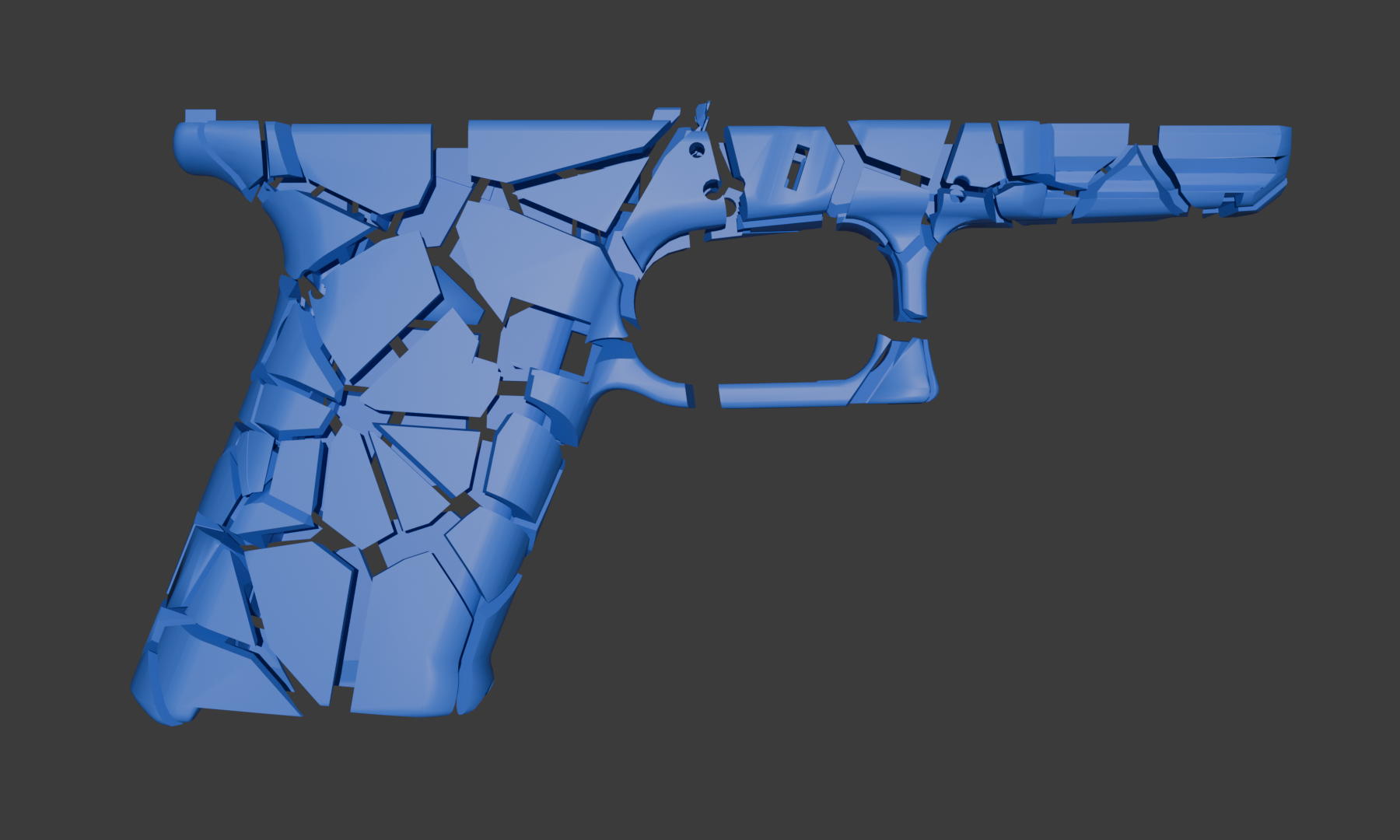}
		\caption{$\beta=100$.}
	\end{subfigure}
 	\caption{Results of fragmentation under different~$\beta$ values.}
        \label{fig:fragmentation}
\end{figure}

\bsub{Simulation Environment:}
This section presents a simulation-based study of fingerprint recovery, which extends the experiments described in the previous section with three significant enhancements.
First, it removes the constraints on the number and orientation of breaks.\footnote{In experiments, breaks are restricted to being orthogonal to the printing direction, and their number is limited by the maximum allowable fragmentation of embedded BRC codeword.}
Second, it introduces greater diversity in model shape and security parameter.
Finally, it evaluates the decoding success rate under the condition of fragment loss.
These enhancements provide more realistic simulations that closely mirror real-world forensic scenarios.

Each simulation is defined by three parameters,~$\alpha$, $\beta$, and~$\rho$.
With a fingerprint information of $128$ bits (the value is chosen based on the length of serial numbers of printers used in experiments), an~$\alpha$-BRC codeword is generated.
The parameter~$\beta$ determines the granularity of fragmentation.
Specifically, a 3D Voronoi diagram from~$\beta$ randomly chosen points within the model mesh is generated, and the model is fractured using the planes separating the Voronoi cells.
This procedure ensures at least~$\beta$ fragments since there are~$\beta$ Voronoi cells, each corresponds to at least one fragment (see fragmentation results in Figure~\ref{fig:fragmentation}).
Finally, the parameter~$\rho$ determines the ratio of fragments hidden from law enforcement.

Simulations are conducted on three models: the FMDA Glock frame (Fig.~\ref{fig:g19}), an AR-15 lower receiver, and the benign 3DBenchy.
For each $\beta\in\{10, 20, \ldots, 100\}$, 
the previously described random fragmentation process is applied to each model, generating~$128$ fragmentation instances.
For each of these~$128$ instances, we conceal a uniformly random~$\rho$ percent of the resulting fragments for each~$\rho \in \{0\%, 25\%, 50\%, 75\%\}$, 
and repeat over~$32$ simulations.
The decoder is then tested to determine whether it can recover the fingerprint from the remaining fragments.

We record the success rate of fingerprint recovery for every possible configuration~$(\alpha,\beta,\rho)$ and every model.
All simulation results are provided in Appendix~\ref{appendix:simulation}, while Table~\ref{tab:simulation-glock-simple} offers a subset of them.
With~$\alpha=8$, we observe an exceptionally high probability of fingerprint recovery even in the scenario when the printed Glock frame is broken into~$100$ fragments with~$75\%$ missing from the law enforcements.

\begin{table}[t]
    \resizebox{\columnwidth}{!}{{\begin{tabular}{|ccccccccc|}
\hline
\multicolumn{1}{|c|}{}        & \multicolumn{4}{c|}{4-BRC, 0.342 mm/bit}          & \multicolumn{4}{c|}{8-BRC, 0.215 mm/bit} \\ \hline
\multicolumn{1}{|c|}{\diagbox[height=1.6em]{$\beta$}{$\rho$}} & 0\% & 25\%  & 50\%  & \multicolumn{1}{c|}{75\%}   & 0\%     & 25\%     & 50\%    & 75\%      \\ \hline
\multicolumn{1}{|c|}{20}      & 100\% & 99.90\%  & 97.05\% & \multicolumn{1}{c|}{48.36\%}  & 100\%     & 100\%      & 100\%     & 95.95\%     \\
\multicolumn{1}{|c|}{40}      & 100\% & 99.98\% & 99.00\%    & \multicolumn{1}{c|}{57.52\%}  & 100\%     & 100\%      & 100\%     & 99.80\%    \\
\multicolumn{1}{|c|}{60}      & 100\% & 99.98\% & 98.80\%  & \multicolumn{1}{c|}{53.00\%}  & 100\%     & 100\%      & 100\%     & 99.95\%     \\
\multicolumn{1}{|c|}{80}      & 100\% & 99.98\% & 98.17\% & \multicolumn{1}{c|}{46.80\%} & 100\%     & 100\%      & 100\%     & 99.98\%     \\
\multicolumn{1}{|c|}{100}     & 100\% & 99.93\% & 96.75\% & \multicolumn{1}{c|}{33.59\%}  & 100\%     & 100\%      & 100\%     & 100\%       \\ \hline
\end{tabular}%
}}
    \caption{\blue{A portion of simulation results on the FMDA Glock frame using 4-BRC and 8-BRC, which demonstrates an exceptionally high probabilities of fingerprint recovery even in the extreme cases.}}
\label{tab:simulation-glock-simple}
\end{table}

\subsection{Code Rate}
The code rate~$r$ of BRC is defined as the ratio between the codeword length~$n$ and the information length~$k$, i.e.,
\begin{equation*}
r=\frac{k}{n}=\frac{l\cdot m -1}{l(m+\ceil{\log m}+4)+a(4m+11)},
\end{equation*}
and plays a critical role in forensic applications. 
Given an object to print and a security parameter~$\alpha$, a higher code rate 
allows embedding more information bits, supporting more advanced forensic functionalities.
These bits may include user IDs, geoposition data, and even the hash of model file.
For example, embedding a unique user ID into a printed object can help trace adversaries using Manufacturing-as-a-Service (MaaS) to make the criminal tool.
Geoposition data at the time of printing can aid law enforcement in tracking adversaries, while hash values act as watermarks for robust IP protection.

\begin{figure}[t]
    \centering
    \scalebox{.7}{
\begin{tikzpicture}
\begin{axis}[
    colormap/viridis,
    xmin=0, xmax=420,
    ymin=1, ymax=15,  
    zmin=0, zmax=250,
    xtick={50,100,150,200,250,300,350,400},
    ytick={2,4,6,8,10,12,14},
    ztick={0,50, 100, 150, 200, 250},
    grid=major,
    title={Minimum object dimension versus~$k$ and $\alpha$},
    xlabel={information length~$k$},
    xlabel style={
        rotate=-9,
        yshift=5pt,
    },
    ylabel={security parameter~$\alpha$},
    ylabel style={
        rotate=38,
        yshift=8pt,
        xshift=-5pt
    },
    zlabel={minimum dimension (mm)},
    zlabel style={
        yshift=-5pt,
        xshift=12,
    },
    ]
    \addplot3 [
        surf,
        shader=faceted,
    ] table {
3 1 5.64
3 2 13.08
3 3 18.84
3 4 29.64
3 5 36.6
3 6 43.56
3 7 50.52
3 8 68.4
3 9 76.68
3 10 84.96
3 11 93.24
3 12 101.52
3 13 109.8
3 14 118.08

11 1 8.88
11 2 14.64
11 3 24.48
11 4 31.44
11 5 38.4
11 6 45.36
11 7 62.28
11 8 70.56
11 9 78.84
11 10 87.12
11 11 95.4
11 12 103.68
11 13 111.96
11 14 120.24

17 1 10.44
17 2 19.32
17 3 26.28
17 4 33.24
17 5 40.2
17 6 56.16
17 7 62.28
17 8 70.56
17 9 78.84
17 10 87.12
17 11 95.4
17 12 103.68
17 13 111.96
17 14 120.24

31 1 14.16
31 2 21.12
31 3 28.08
31 4 35.04
31 5 50.04
31 6 58.32
31 7 66.6
31 8 74.88
31 9 83.16
31 10 91.44
31 11 99.72
31 12 108.0
31 13 132.84
31 14 139.92

39 1 15.96
39 2 22.92
39 3 29.88
39 4 43.92
39 5 50.04
39 6 58.32
39 7 66.6
39 8 74.88
39 9 83.16
39 10 91.44
39 11 99.72
39 12 108.0
39 13 132.84
39 14 142.32

47 1 17.76
47 2 24.72
47 3 37.8
47 4 43.92
47 5 52.2
47 6 60.48
47 7 68.76
47 8 77.04
47 9 85.32
47 10 93.6
47 11 101.88
47 12 125.76
47 13 132.84
47 14 142.32

55 1 19.56
55 2 31.68
55 3 37.8
55 4 46.08
55 5 54.36
55 6 62.64
55 7 70.92
55 8 79.2
55 9 87.48
55 10 95.76
55 11 118.68
55 12 125.76
55 13 135.24
55 14 144.72

79 1 25.56
79 2 33.84
79 3 42.12
79 4 50.4
79 5 58.68
79 6 66.96
79 7 75.24
79 8 83.52
79 9 104.52
79 10 111.6
79 11 121.08
79 12 130.56
79 13 140.04
79 14 149.52

89 1 27.72
89 2 36.0
89 3 44.28
89 4 52.56
89 5 60.84
89 6 69.12
89 7 77.4
89 8 97.44
89 9 104.52
89 10 114.0
89 11 123.48
89 12 132.96
89 13 142.44
89 14 151.92

99 1 29.88
99 2 38.16
99 3 46.44
99 4 54.72
99 5 63.0
99 6 71.28
99 7 90.36
99 8 97.44
99 9 106.92
99 10 116.4
99 11 125.88
99 12 135.36
99 13 144.84
99 14 154.32

109 1 32.04
109 2 40.32
109 3 48.6
109 4 56.88
109 5 65.16
109 6 83.28
109 7 90.36
109 8 99.84
109 9 109.32
109 10 118.8
109 11 128.28
109 12 137.76
109 13 147.24
109 14 156.72

119 1 34.2
119 2 42.48
119 3 50.76
119 4 59.04
119 5 76.2
119 6 83.28
119 7 90.36
119 8 99.84
119 9 109.32
119 10 118.8
119 11 128.28
119 12 137.76
119 13 147.24
119 14 156.72

129 1 36.36
129 2 44.64
129 3 52.92
129 4 69.12
129 5 76.2
129 6 83.28
129 7 92.76
129 8 102.24
129 9 111.72
129 10 121.2
129 11 130.68
129 12 140.16
129 13 149.64
129 14 159.12

139 1 38.52
139 2 46.8
139 3 62.04
139 4 69.12
139 5 76.2
139 6 85.68
139 7 95.16
139 8 104.64
139 9 114.12
139 10 123.6
139 11 133.08
139 12 142.56
139 13 152.04
139 14 161.52

149 1 40.68
149 2 54.96
149 3 62.04
149 4 69.12
149 5 78.6
149 6 88.08
149 7 97.56
149 8 107.04
149 9 116.52
149 10 126.0
149 11 135.48
149 12 144.96
149 13 154.44
149 14 163.92

191 1 47.88
191 2 57.36
191 3 66.84
191 4 76.32
191 5 85.8
191 6 95.28
191 7 104.76
191 8 114.24
191 9 123.72
191 10 133.2
191 11 142.68
191 12 152.16
191 13 161.64
191 14 171.12

203 1 50.28
203 2 59.76
203 3 69.24
203 4 78.72
203 5 88.2
203 6 97.68
203 7 107.16
203 8 116.64
203 9 126.12
203 10 135.6
203 11 145.08
203 12 154.56
203 13 164.04
203 14 173.52

215 1 52.68
215 2 62.16
215 3 71.64
215 4 81.12
215 5 90.6
215 6 100.08
215 7 109.56
215 8 119.04
215 9 128.52
215 10 138.0
215 11 147.48
215 12 156.96
215 13 166.44
215 14 175.92

227 1 55.08
227 2 64.56
227 3 74.04
227 4 83.52
227 5 93.0
227 6 102.48
227 7 111.96
227 8 121.44
227 9 130.92
227 10 140.4
227 11 149.88
227 12 159.36
227 13 168.84
227 14 199.68

239 1 57.48
239 2 66.96
239 3 76.44
239 4 85.92
239 5 95.4
239 6 104.88
239 7 114.36
239 8 123.84
239 9 133.32
239 10 142.8
239 11 152.28
239 12 161.76
239 13 191.64
239 14 199.68

251 1 59.88
251 2 69.36
251 3 78.84
251 4 88.32
251 5 97.8
251 6 107.28
251 7 116.76
251 8 126.24
251 9 135.72
251 10 145.2
251 11 154.68
251 12 183.6
251 13 191.64
251 14 199.68

263 1 62.28
263 2 71.76
263 3 81.24
263 4 90.72
263 5 100.2
263 6 109.68
263 7 119.16
263 8 128.64
263 9 138.12
263 10 147.6
263 11 175.56
263 12 183.6
263 13 191.64
263 14 199.68

275 1 64.68
275 2 74.16
275 3 83.64
275 4 93.12
275 5 102.6
275 6 112.08
275 7 121.56
275 8 131.04
275 9 140.52
275 10 167.52
275 11 175.56
275 12 183.6
275 13 191.64
275 14 202.32

287 1 67.08
287 2 76.56
287 3 86.04
287 4 95.52
287 5 105.0
287 6 114.48
287 7 123.96
287 8 133.44
287 9 159.48
287 10 167.52
287 11 175.56
287 12 183.6
287 13 194.28
287 14 204.96

299 1 69.48
299 2 78.96
299 3 88.44
299 4 97.92
299 5 107.4
299 6 116.88
299 7 126.36
299 8 151.44
299 9 159.48
299 10 167.52
299 11 175.56
299 12 186.24
299 13 196.92
299 14 207.6

311 1 71.88
311 2 81.36
311 3 90.84
311 4 100.32
311 5 109.8
311 6 119.28
311 7 143.4
311 8 151.44
311 9 159.48
311 10 167.52
311 11 178.2
311 12 188.88
311 13 199.56
311 14 210.24

323 1 74.28
323 2 83.76
323 3 93.24
323 4 102.72
323 5 112.2
323 6 135.36
323 7 143.4
323 8 151.44
323 9 159.48
323 10 170.16
323 11 180.84
323 12 191.52
323 13 202.2
323 14 212.88

335 1 76.68
335 2 86.16
335 3 95.64
335 4 105.12
335 5 127.32
335 6 135.36
335 7 143.4
335 8 151.44
335 9 159.48
335 10 170.16
335 11 180.84
335 12 191.52
335 13 202.2
335 14 212.88

347 1 79.08
347 2 88.56
347 3 98.04
347 4 119.28
347 5 127.32
347 6 135.36
347 7 143.4
347 8 151.44
347 9 162.12
347 10 172.8
347 11 183.48
347 12 194.16
347 13 204.84
347 14 215.52

359 1 81.48
359 2 90.96
359 3 111.24
359 4 119.28
359 5 127.32
359 6 135.36
359 7 143.4
359 8 154.08
359 9 164.76
359 10 175.44
359 11 186.12
359 12 196.8
359 13 207.48
359 14 218.16

371 1 83.88
371 2 103.2
371 3 111.24
371 4 119.28
371 5 127.32
371 6 135.36
371 7 146.04
371 8 156.72
371 9 167.4
371 10 178.08
371 11 188.76
371 12 199.44
371 13 210.12
371 14 220.8

    };
\end{axis}
\end{tikzpicture}

 }
    \caption{An illustration of minimum object dimension that allows for BRC codeword embedding, in which a bit is represented by~$0.12$ mm in object height, with respect to information length~$k$ and security parameter~$\alpha$.
    }
    \label{fig:minsize}
\end{figure}
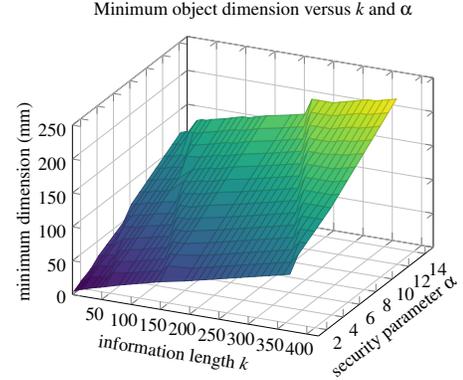

\begin{figure}[t]
	\begin{subfigure}{0.33\textwidth}
            \centering
		\includegraphics[height=2cm]{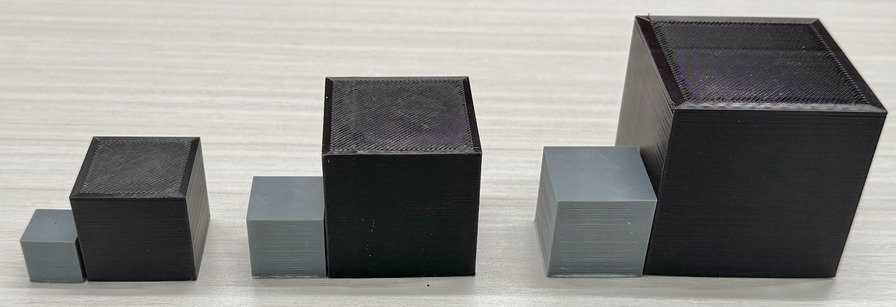}
		\caption{}\label{fig:cubes}
	\end{subfigure}
	\begin{subfigure}{0.14\textwidth}
            \centering
		\includegraphics[height=2cm]{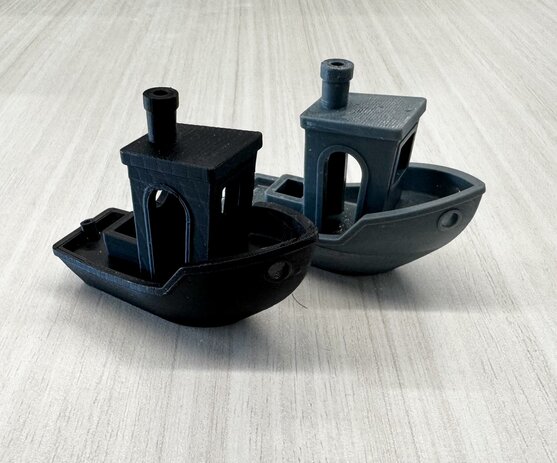}
		\caption{}\label{fig:benchys}
	\end{subfigure}
 	   \caption{\blue{(a) From a single information word of 39 bits, we generate 1-BRC, 2-BRC, and 3-BRC codewords with lengths of 133, 191, and 249 bits, respectively. For each codeword, we calculate the minimum required dimensions based on the embedding parameters described in Section~\ref{section:bit-embedding}. Using these dimensions, we print information-bearing cubes with both the Ender 3 and Mars 4 printers. The resulting cubes have side lengths of 15.96 mm, 31.92 mm, 22.92 mm, 45.84 mm, 29.88 mm, and 59.76 mm. (b) 3DBenchy models printed for surface roughness evaluation.}}
\end{figure}

Conversely, given a fixed model and specific fingerprinting information, a higher code rate permits using a larger security parameter~$\alpha$, increasing resilience to breaks and fragment losses.
Alternatively, when embedding fingerprints of a specific length with a given security parameter, a higher code rate reduces the required object dimension along the printing direction, and hence broadens the applicability of SIDE.
To this end, we visualized the minimum object dimensions required to embed BRC codewords in Figure~\ref{fig:minsize}, and printed~$6$ cubes with BRC codeword embedded to validate the feasibility of embedding; each is of the minimum dimension corresponding to the information length and security parameter (Figure~\ref{fig:cubes}).

\begin{figure}[t]
    \centering
    \scalebox{0.82}{\begin{tikzpicture}
\begin{axis}[
    title={code rate~$r$ vs.~$k,\alpha$},
    xlabel={information length~$k$ },
    ylabel={code rate~$r=k/n$},
    xmin=0, xmax=400,
    ymin=.05, ymax=.6,  
    xtick={50,100,150,200,250,300,350,400},
    ytick={0.1,0.2,0.3,0.4,0.5,0.6},
    legend pos=north west,
    ymajorgrids=false,
    legend style={nodes={scale=0.8, transform shape}}
]

\addplot[
    color=blue,
    mark=square,
    ]
    coordinates {
    (191, 0.48111)(203, 0.48681)(215, 0.49199)(227, 0.49672)(239, 0.50105)(251, 0.50503)(263, 0.5087)(275, 0.5121)(287, 0.51526)(299, 0.5182)(311, 0.52094)(323, 0.5235)(335, 0.5259)(347, 0.52816)(359, 0.53028)(371, 0.53228)
    };  
\addlegendentry{$a=1$}

\addplot[
    color=red,
    mark=triangle,
    ]
    coordinates {
    (179, 0.394)(191, 0.403)(203, 0.411)(215, 0.418)(227, 0.425)(239, 0.431)(251, 0.437)(263, 0.443)(275, 0.448)(287, 0.453)(299, 0.457)(311, 0.461)(323, 0.465)(335, 0.469)(347, 0.473)
    };
\addlegendentry{$a=2$} 

\addplot[
    color=green,
    mark=o,
    ]
    coordinates {
    (167, 0.327)(179, 0.337)(191, 0.347)(203, 0.356)(215, 0.364)(227, 0.372)(239, 0.379)(251, 0.386)(263, 0.392)(275, 0.398)(287, 0.404)(299, 0.409)(311, 0.414)(323, 0.419)(335, 0.424)
    };
\addlegendentry{$a=3$}

\addplot[
    color=yellow,
    mark=+,
    ]
    coordinates {
    (155, 0.273)(167, 0.284)(179, 0.294)(191, 0.304)(203, 0.313)(215, 0.322)(227, 0.33)(239, 0.338)(251, 0.345)(263, 0.352)(275, 0.358)(287, 0.364)(299, 0.37)(311, 0.376)(323, 0.381)
    };
\addlegendentry{$a=4$}

\addplot[
    color=black,
    mark=star,
    ]
    coordinates {
    (143, 0.229)(155, 0.24)(167, 0.251)(179, 0.261)(191, 0.271)(203, 0.28)(215, 0.289)(227, 0.297)(239, 0.304)(251, 0.312)(263, 0.319)(275, 0.325)(287, 0.332)(299, 0.338)(311, 0.344)
    };
\addlegendentry{$a=5$}

\addplot[
    color=violet,
    mark=x,
    ]
    coordinates {
    (131, 0.192)(143, 0.204)(155, 0.215)(167, 0.225)(179, 0.235)(191, 0.244)(203, 0.253)(215, 0.262)(227, 0.27)(239, 0.277)(251, 0.285)(263, 0.292)(275, 0.298)(287, 0.305)(299, 0.311)
    };
\addlegendentry{$a=6$}

\addplot[
    color=cyan,
    mark=+,
    ]
    coordinates {
    (119, 0.161)(131, 0.173)(143, 0.184)(155, 0.194)(167, 0.204)(179, 0.213)(191, 0.222)(203, 0.231)(215, 0.239)(227, 0.247)(239, 0.255)(251, 0.262)(263, 0.269)(275, 0.275)(287, 0.282)
    };
\addlegendentry{$a=7$}

\addplot[
    color=magenta,
    mark=diamond,
    ]
    coordinates {
    (107, 0.134)(119, 0.146)(131, 0.157)(143, 0.167)(155, 0.177)(167, 0.186)(179, 0.195)(191, 0.204)(203, 0.212)(215, 0.22)(227, 0.228)(239, 0.235)(251, 0.242)(263, 0.249)(275, 0.256)
    };
\addlegendentry{$a=8$}

\addplot[
    color=gray,
    mark=pentagon,
    ]
    coordinates {
    (95, 0.11137)(107, 0.12257)(119, 0.13326)(131, 0.14348)(143, 0.15327)(155, 0.16264)(167, 0.17163)(179, 0.18026)(191, 0.18855)(203, 0.19652)(215, 0.20418)(227, 0.21156)(239, 0.21866)(251, 0.22552)(263, 0.23213)(275, 0.23851)
    };
\addlegendentry{$a=9$} 

\addplot[
    color=lightgray,
    mark=oplus,
    ]
    coordinates {
    (83, 0.09121)(95, 0.10215)(107, 0.11263)(119, 0.12268)(131, 0.13232)(143, 0.14158)(155, 0.15049)(167, 0.15905)(179, 0.16729)(191, 0.17523)(203, 0.18288)(215, 0.19027)(227, 0.19739)(239, 0.20427)(251, 0.21092)(263, 0.21736)
    };
\addlegendentry{$a=10$}

\addplot [
    domain=0:400, 
    samples=10, 
    color=blue,
    ]
    {0.5};

\addplot [
    domain=0:400, 
    samples=10, 
    color=red,
    ]
    {1/3};

\addplot [
    domain=0:400, 
    samples=10, 
    color=green,
    ]
    {0.25};

\addplot [
    domain=0:400, 
    samples=10, 
    color=yellow,
    ]
    {0.2};

\addplot [
    domain=0:400, 
    samples=10, 
    color=black,
    ]
    {1/6};

\addplot [
    domain=0:400, 
    samples=10, 
    color=violet,
    ]
    {1/7};

\addplot [
    domain=0:400, 
    samples=10, 
    color=cyan,
    ]
    {1/8};

\addplot [
    domain=0:400, 
    samples=10, 
    color=magenta,
    ]
    {1/9};

\addplot [
    domain=0:400, 
    samples=10, 
    color=gray,
    ]
    {1/10};
    
\addplot [
    domain=0:400, 
    samples=10, 
    color=lightgray,
    ]
    {1/11};

\end{axis}
\end{tikzpicture}}
    \caption{An illustration of the code rate~$r$ verses the information length~$k$ and security parameter~$\alpha$. The horizontal lines serve as upper bounds on the code rates of a traditional CPC-based scheme (Section~\ref{section:coding-methods}).
    The curved lines represent the code rate of our scheme (Section~\ref{section:BRC}), color coded by~$\alpha$.
    }
    \label{fig:tkn}
\end{figure}
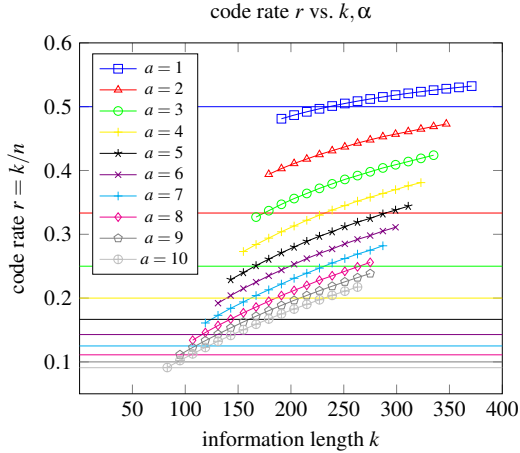

In Figure~\ref{fig:tkn}, we demonstrate how the code rate~$r$ is affected by the information length~$k=m\cdot l-1$ and the security threshold~$t\in[10]$.\footnote{Since we fix~$m=12$, the maximum~$l$ we can have is~$2^{(12-2)/2}-t=32-t$.
Therefore, the set of values of~$k$ is different with different~$t$.}
Meanwhile, as we will introduce in Section~\ref{section:coding-methods}, the method of using cyclically permutable codes (CPC) for the~$t$-break-resilient property requires repeating the CPC codeword~$t+1$ times, and hence the code rate is at most~$1/(t+1)$.
We also mark this value in the figure for every~$t$ in the same color.
In the majority of cases, BRC outperforms CPC-based scheme in terms of the code rate.

\subsection{Stealthy Embedding}
The primary goal of stealthy embedding is to minimize detectability, ensuring that the fingerprint remains hidden from adversaries while preserving its readability by forensic tools.
Hence, the feasibility of a stealthy embedding method is determined by its~\emph{stealthiness} and~\emph{readability}.
 
To evaluate stealthiness, we conducted experiments using both FDM and SLA printers.
Specifically, we use the stealthy embedding method introduced in Section~\ref{section:bit-embedding} with parameter~$(y,\epsilon)=(0.12,0.04)$ for the FDM printer and~$(y,\epsilon)=(0.6,0.02)$ for the SLA printer.
Results indicate that the differences in layer thicknesses were invisible to naked eyes under normal lighting conditions, unless the object is observed in certain angles.
We further measured surface roughness with our SD-OCT system.
The results, shown in Table~\ref{tab:RMS}, reveal a higher RMS value in prints with no embedded bits, but lower than prints using normal bit embedding. 

Readability refers to the accuracy of bit extraction.
With the extraction method introduced in Section~\ref{section:bit-extraction} and the SLA printer with parameters~$(y,\epsilon) = (0.08,0.04)$, we successfully extracted embedded bits using SD-OCT system.
Yet, we observed a trade-off between stealthiness and readability: while reducing~$\epsilon$ enhances stealthiness, it increases the demand for both high-resolution printer and extraction tools.

\subsection{Impacts on Print Quality}

\begin{figure}[t]
	\centering
	\begin{subfigure}{0.23\textwidth}
		\includegraphics[height=1.47cm]{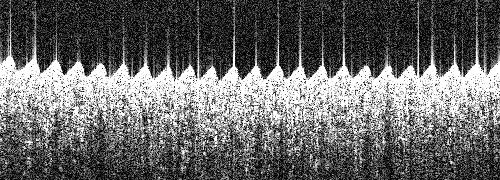}
		\caption{SLA, no embedding.}
	\end{subfigure}
        \hfill
	\begin{subfigure}{0.23\textwidth}
		\includegraphics[height=1.47cm]{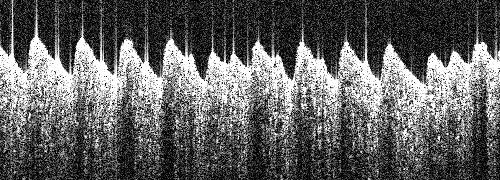}
		\caption{SLA, stealthy embedding.}
	\end{subfigure}\\
    \centering
	\begin{subfigure}{0.23\textwidth}
		\includegraphics[height=1.47cm]{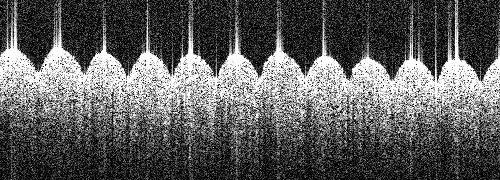}
		\caption{FDM, no embedding.}
	\end{subfigure}
        \hfill
	\begin{subfigure}{0.23\textwidth}
		\includegraphics[height=1.47cm]{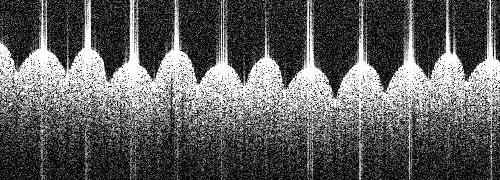}
		\caption{FDM, stealthy embedding.}
	\end{subfigure}
 	\caption{OCT scans on different materials and embeddings.}
        \label{fig:ost-scans}
\end{figure}

Due to the inherently discrete, layer-by-layer nature of the additive manufacturing process, surface roughness is a common characteristic of 3D-printed objects, often necessitating post-processing techniques such as sanding or filing.
Yet, SIDE's bit embedding requires to vary layer thickness, which can potentially increase surface roughness and expose the presence of the embedded bits to the adversary.

To quantify this impact, we printed 3DBenchy with normal embedding, stealthy embedding, and no embedding (uniform~$0.12$ and~$0.06$ mm layer thicknesses for FDM and SLA, respectively) with both FDM and SLA printers.
For each printed model, we used OCT to capture surface height deviations at five randomly chosen points of the flat part of its surface.
With the collected data, we computed and averaged the root mean square (RMS) roughness value, which is a widely recognized metric for surface roughness that quantifies the root mean square of surface height deviations from the mean surface height. The results are provided in Table~\ref{tab:RMS}, which quantify the impacts on print quality.

\begin{table}[t]
\centering
\resizebox{\columnwidth}{!}{\begin{tabular}{c|c|c|c}
\hline
              & No Embedding & Stealthy & Normal         \\ \hline
Ender 3 (FDM) & 10.741 $\mu$m      & 13.267 $\mu$m  & 14.719 $\mu$m  \\ \hline
Mars 4 (SLA)  & 7.260  $\mu$m      & 11.275 $\mu$m  & 12.918  $\mu$m \\ \hline
\end{tabular}%
}
\caption{The RMS values with different embedding methods and materials.}
\label{tab:RMS}
\end{table}

\subsection{Impact of Printer Imperfection}
In this section we briefly discuss the effect that printer imperfections might have on the ability to embed information.
Recall that SIDE embeds bits by varying layer thickness, and as a result, imperfections in 3D printers---particularly inaccuracies in the $z$-axis movement---can have negative effects on the embedding process.
The $z$-axis movements are typically controlled by a stepper motor, which converts its rotational motion (i.e., discrete steps) into linear motion along the printing direction; the ratio between them is reflected in~\emph{z-step} value (e.g.,~$0.04$ mm/step).
Calibration of a printer involves matching this parameter in the printer firmware with the actual~\emph{z-step} value.

A perfectly manufactured printer would have a uniform $z$-step, i.e., the nozzle displacement along the~$z$-axis, 
which is triggered by one microstep, 
is uniform across the entire range.
Yet, due to manufacturing errors, individual steps may lead to different nozzle displacements during printing, and the actual layer thickness may not align with the desired value.

SIDE's embedding method is designed to tolerate these imperfections.
For normal embedding, layers with designed thicknesses of~$x$, $2x$, and~$3x$ may be in the range of $
[(1-\delta)\cdot x,(1+\delta)\cdot x], [(1-\delta)\cdot 2x,(1+\delta)\cdot 2x]$, and~$[(1-\delta)\cdot 3x,(1+\delta)\cdot 3x]$, where~$\delta$ upper bounds the magnitude of errors in each step.
To avoid confusion in the reading of bits from layer thickness, these ranges must not overlap, requiring~$\delta < 0.2$; this is an extremely low standards, considering that the common stepper motors used in 3D printers (e.g., NEMA 17) usually have~$\delta=0.05$.

To evaluate SIDE's robustness under these imperfections, a cuboid was printed using an FDM printer with noise added to the $z$-step value in the firmware; this simulates the behaviors of a poorly manufactured printer.
Despite these imperfections, bits were successfully extracted, demonstrating the method's reliability even in the extreme real-world conditions.

\subsection{Impact on Printing from TEE Protection}
\begin{table}[t]
\resizebox{\columnwidth}{!}{\begin{tabular}{c|cc|cc|clcl}
\hline
         & \multicolumn{2}{c|}{Peak Heap Usage} & \multicolumn{2}{c|}{Print Time} & \multicolumn{4}{c}{Firmware Execution Time}               \\ \hline
Object   & w/ TEE            & w/o TEE          & w/ TEE (s)         & w/o TEE (s)       & \multicolumn{2}{c}{w/ TEE} & \multicolumn{2}{c}{w/o TEE}  \\ \hline
Glock    & 1.475 MB          & 9.841 MB         & 53344.201             & 53344.103        & \multicolumn{2}{c}{163.85 s}       & \multicolumn{2}{c}{163.85 s} \\
AR-15    & 1.305 MB          & 8.438 MB         & 76831.497             & 76831.397        & \multicolumn{2}{c}{209.69 s}       & \multicolumn{2}{c}{209.69 s} \\
3DBenchy & 0.965 MB          & 9.979 MB         & 11828.314              & 11828.241        & \multicolumn{2}{c}{34.90 s}       & \multicolumn{2}{c}{34.90 s}  \\ \hline
\end{tabular}%
}
\caption{System Performance of SIDE.}
\label{tab:sysperf}
\end{table}

SIDE prevents attackers from exploiting vulnerabilities in untrusted software by leveraging TEE to isolate SIDE from untrusted components. To evaluate the impact of TEE protection on 3D printing, we measure (1) the execution delay of SIDE's software components, (2) the execution delay and memory overhead of the whole 3D printing procedures, and (3) the impact on the quality of printed objects with and without SIDE on Glock frame, AR-15 lower receiver, and 3DBenchy.

\bsub{Software Components Execution and 3D Printing Delay: }
To measure the execution delay of SIDE's main components, including the codec, G-Code generation (involving both layer and toolpath generation), and firmware, we record timestamps at the start and end of each component's execution and calculate the average delays over~$10$ printing processes.
To evaluate the overall 3D printing delay, we record timestamps marking the start and completion of the printing application execution and calculate the difference to determine the delay.
As shown in Table~\ref{tab:sysperf}, SIDE introduces no runtime overhead for the execution delays of individual software components under protection of TEE, as the printing binaries remain the same on the same architecture regardless of the CPU security state.
However, SIDE reduces peak memory usage by employing the progressive slicing strategy, which introduces multiple context switches between the REE and TEE. This results in additional execution delay for the overall printing process.
Nonetheless, when compared to the delay inherent in physical printing, this runtime overhead is negligible.

\bsub{Memory Overhead:}
To measure the effectiveness of SIDE in automatically splitting 3D object models to reduce heap consumption, we instrument dynamic memory allocation and deallocation APIs within the 3D printing software to monitor the size of dynamically allocated memory, both with and without SIDE.
As shown in Table~\ref{tab:sysperf}, SIDE reduces the peak heap memory usage to 14.99\%, 15.47\%, and 9.67\% on Glock, AR-15, and 3DBenchy respectively through progressive slicing, effectively addressing the memory limitations of TEE.

\bsub{Printed Object Quality: }
To measure the quality impact on printed object from TEE implementation, we calculate the root mean square (RMS) roughness of printed object with observed height from optical coherence tomography with and without TEE implementation.
We observe that the difference in the roughness of both objects is negligible.

\section{Security Analysis and Discussion}
This section analyzes the security of SIDE against various potential attacks.

\bsub{Excessive breaking and hiding:}
In order to jeopardize fingerprint extraction, the adversary may attempt to compromise the availability of embedded information by excessively breaking the printed tool and hide a great amount of fragments.
However, as shown in the simulation results in Appendix~\ref{appendix:simulation}, SIDE provide an exceptionally high success rate of fingerprinting recovery even in the extreme case that the tool is broken into~$100$ pieces, with~$75\%$ of them being missing from the decoder.

This is attributed to the break-resiliency and loss-tolerance properties of~$\alpha$-BRC.
First, as stated in Remark~\ref{remark:type-of-breaks} and Section~\ref{section:preprocessing}, fragments can be reassembled if they retain sufficient overlapping bits.
Thus, a break is repairable in the preprocessing stage unless it is perpendicular, or nearly perpendicular, to the printing direction.
Additionally, concealing fragments does not necessarily result in the loss of information, as their content may also exist in confiscated fragments.
Finally, the breaks and missing bits that cannot be fixed in the preprocessing stage are addressed in BRC decoding (Section~\ref{section:decoding}).

\bsub{Forging Attacks:} Attackers may attempt forging attacks by embedding incorrect information into a printed object to impersonate another printer.
We propose two tiered defense against such attacks.
First, in order to impersonate another printer, one would have to breach the TEE of one's own printer, contradicting our security assupmtion.
Second, even if the TEE is broken successfully, impersonation can still be prevented via exploiting intrinsic printer properties as follows. 
It was shown in the literature that much like firearms, 3D printers carry a unique signature that is manifested in various minor defects in the resulting prints~\cite{li2018printracker,gao2021thermotag}. 
Hence, SIDE can embed a hash of these imperfections inside the object. 
Then, forging attacks would fail since the de-facto defects of the object would not match the hash of the printer's imperfection.
However, manipulating the embedding mechanisms in SIDE requires significant additional effort, such as breaching the TEE-protected 3D printing software or employing specialized hardware and expertise to modify the object’s surface.
Thus, SIDE raises the bar for this attack.
Furthermore, SIDE can authenticate the print by embedding the hash of intrinsic printer properties, such as manufacturing imperfections~\cite{li2018printracker} or thermodynamic characteristics~\cite{gao2021thermotag}, enabling detection of mismatches between the print and its originating printer.

\bsub{Hardware Attacks: }
Adversaries may attempt hardware attacks, such as replacing the control board or injecting signals between the control board and printers, to bypass fingerprint embedding. While SIDE does not defend against these hardware attacks, which require specialized expertise and significant cost, establishing secure communication channels and authentication mechanisms between the control board and printer component controllers can help mitigate such attacks, further raising the barriers to bypassing 3D fingerprinting.

\bsub{Surface Altering Attack:}
Attackers may attempt to alter the surface of 3D prints to tamping with the 3D fingerprinting information. SIDE is inherently immune to this type of attack by its design. Recall that SIDE embeds information by altering physical elements on the printing direction (e.g., layer thickness), which is an intrinsic property of the printed object and cannot be altered by post processings on the surface. Indeed, the reading of bits does not require the surface of the fragment to be free of adversarial tampering, as less-economical solutions, such as an industrial CT scanner, can be used to infer the layer thickness via tomographic analysis.

\section{Related Work}

\subsection{Existing 3D Fingerprinting Methods}\label{section:existing}
Several methods for embedding bits in 3D-printed objects have recently been proposed in the literature.
These technologies allow the printer to vary either the orientation of the nozzle, the thickness of the layer, or the printing speed.
Within reasonable bounds, varying either of those has a marginal effect on the functionality of the object.
By varying layer thickness, for example, the printer can embed a~$0$ by printing a layer that is slightly thinner, and a~$1$ by printing a layer that is slightly thicker, than some reference thickness.
By varying the orientation of the nozzle, bits can be embedded by the relative orientation of adjacent layers; for example, if two adjacent layers are oriented similarly, the embedded bit is~$0$, and otherwise it is~$1$. Both methods are illustrated in Figure~\ref{fig:embedding}.

\begin{figure}[t]
	\centering
	\begin{subfigure}{0.23\textwidth}
		\centering
		\includegraphics[height=2.6cm]{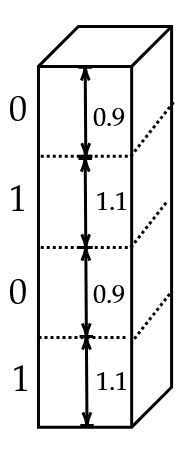}
		\caption{}\label{figure:additionalExamplesA}
	\end{subfigure}
        \hfill
	\begin{subfigure}{0.23\textwidth}
		\centering
		\includegraphics[height=2.6cm]{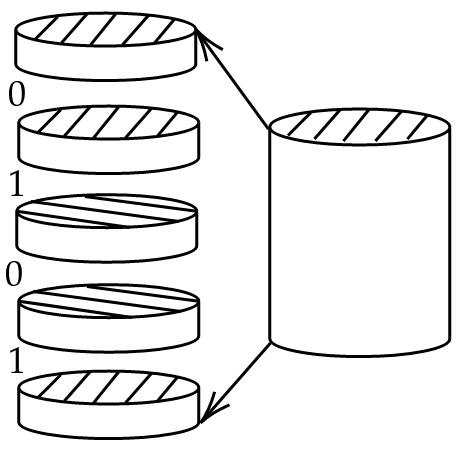}
		\caption{}\label{figure:additionalExamplesB}
	\end{subfigure}	
	\caption{Two possible methods for embedding bits in a 3D printed object with little to no effect on functionality:
                (a) Embedding by layer thickness; thicker layers represent~$1$ and thinner layers represent~$0$.
                (b) Embedding bits using the orientation of adjacent layers; if two adjacent layers are oriented similarly, it is a~$0$, and if oriented differently, it is a~$1$.
                Both left and right figures contain the bits~$0101$.}\label{fig:embedding}
\end{figure} 

Similar ideas have been implemented successfully in several recent works.
Delmotte~\textit{et al.} vary the thickness of each layer across several adjacent layers to create a matrix of bits that is visible to the naked eye on the surface of the object~\cite{delmotte2019blind}.
Parity bits were then added to resolve reliability issues in some cases, and additional noise patterns were discussed, such as orientation issues and sanding.
In the method LayerCode~\cite{maia2019layercode}, variations in color and thickness were used to embed a barcode on the surface of the printed object, that can be retrieved using a smartphone camera.
An orientation-based method has been implemented in~\cite{hou20153d}, where the authors print a reference layer that is circularly grooved by a sequence with low auto-correlation.
Data is embedded in all other layers by the respective angle of the layer to the reference layer; this enables encoding with alphabet size larger than two.

Other creative ideas have been explored, including embedding information-carrying cavities within the object~\cite{willis2013infrastructs,li2017aircode}, water-marking the 3D-mesh of the surface of the object~\cite{cho2006oblivious}, inserting RFID tags~\cite{voris2017three}, inserting a series of notches which create an acoustic barcode when tapped~\cite{harrison2012acoustic}, etc.
In the data extraction phase, most existing methods rely on an RGB camera, a 3D scanner, or an ordinary scanner.
Future technologies however, such as the ones in Figure~\ref{fig:embedding}, might require an industrial CT scanner.
However, none of these approaches is suitable for forensic applications. 
First, they all implicitly assume that a mechanism ensuring correct information embedding is in place.
Such a mechanism is crucial since in most scenarios the adversary owns the printer and/or the file, and might potentially remove the embedded bits altogether.
Second, none of the methods is provably resilient to adversarial tampering; they can be easily breached by an adaptive adversary that can scrape the object or break it apart.
 
PrinTracker~\cite{li2018printracker} represents a different line of works, which apply machine-learning based methods to identify the intrinsic fingerprint resulting from printer hardware imperfections.
Our proposed solution allows for embedding arbitrary bits, which enables a variety of forensic applications, including but not limited to fingerprinting.
For example, certain applications may require the embedding of timestamp or geoposition of the printing into the printed object in order to check whether the printer has been misused in unauthorized time and location.

\subsection{Coding Methods} \label{section:coding-methods}
Secure information extraction is essentially a problem of communicating in the presence of (potentially adversarial) noise.
Such problems are studied in the information-theoretic literature under the title of coding theory~\cite{roth2006coding}.
A typical problem setup in coding theory includes a sender (e.g., the printer), which would like to send a message (e.g., printer ID) to a receiver (e.g., law-enforcement). 
This must be accomplished successfully even if an adversary (e.g., a criminal) injects adversarial noise that is limited by some security parameter.

Coding problems of information extraction from fragments have been previously studied in the literature, motivated by applications in distributed systems~\cite{wang2021low,wang2022breaking,wang2022all,wang2024encoding}, and DNA storage.
In particular, several variations of the \emph{torn-paper channel} were studied in~\cite{shomorony2020communicating,shomorony2021torn,ravi2021capacity,bar2023adversarial}, where~\cite{shomorony2020communicating,shomorony2021torn,ravi2021capacity} focused on a probabilistic error model which is incompatible with our adversarial setting, and~\cite{bar2023adversarial} studied an adversarial model in which fragment are restricted in length.
SIDE complements these probabilistic-based solutions by relaxing the assumptions.
In fact, we assume an adaptive adversary who is fully aware of the coding scheme, and is constrained only by the parameters~$t$ and~$s$.
Consequently, previous methods fail to ensure secure information embedding facing such an adversary, as they can be exploited by strategically selecting break locations and concealing specific fragments.

Similarly to our scenario, a video watermarking solution for IP protection~\cite{kuribayashi2006generate} employs \emph{cyclically permutable codes} (CPC).
As the name suggests, CPC codewords are cyclically distinct, i.e., one cannot obtain a codeword by cyclically shifting another codeword.
In CPC-based solutions, the watermark is encoded to a CPC codeword, and then iteratively spread over consecutive video frames.
Due to its cyclically distinct feature, the embedded CPC codeword can be obtained from every video clip that has more frames than the codeword length.
Although this method may serve as a simple solution to the 3D fingerprinting problem, it requires to repeat the CPC codeword at least~$t+1$ times to guarantee the existence of such a video clip in every possible way to cut the video~$t$ times.
Hence, it leads to a code rate of at most~$1/(t+1)$, i.e.,~$k(t+1)$ bits needs to be actually embedded in the video to represent~$k$ bits of information, which hinders its applicability in real-world scenario; see Figure~\ref{fig:tkn} for rate comparison.

Most closely related, \cite{wang2024break} studied $t$-break codes, which concerns the recovery of information from~\emph{arbitrarily} broken codewords, which is fundamentally different from probabilistic error models mentioned earlier.
They provided a theoretical analysis of the fundamental limits, and an (almost) matching code construction.
However, the scheme described in~\cite{wang2024break} involves random encoding, is only effective for a very large number of embedded bits, and only tolerates a small amount of lost bits.
SIDE complements these works by offering a coding scheme with tolerance to fragment loss, explicit and efficient encoding function, and high code rate.

\begin{table}[!t]
\scalebox{.64}{
\begin{tabular}{c|c|c|c}
\hline
\textbf{Coding Method}       & \textbf{Model} & \textbf{Minimum Length of Fragments} & \textbf{Loss Tolerance} \\ \hline
Shomorony et al.~\cite{shomorony2020communicating,shomorony2021torn} & probability    & Not Assumed                          & No                      \\ \hline
Ravi et al.~\cite{ravi2021capacity}        & probability    & Not Assumed                          & Yes                     \\ \hline
Bar-Lev et al.~\cite{bar2023adversarial}      & adversarial    & Assumed                              & No                      \\ \hline
Kuribayashi et al.~\cite{kuribayashi2006generate}      & adversarial    & Not Assumed                          & Yes                      \\ \hline
Wang et al.~\cite{wang2024break}       & adversarial    & Not Assumed                          & No                      \\ \hline
$\alpha$-BRC                         & adversarial    & Not Assumed                          & Yes                     \\ \hline
\end{tabular}
}
\caption{Coding Method Comparison}
\end{table}

\section{Conclusion}\label{section:conclusion}

This paper introduces SIDE, a forensic fingerprinting mechanism for 3D prints that ensures secure information embedding and extraction against adversarial threats.
To enable fingerprint recovery despite malicious breaks and fragment hiding, SIDE employs break-resilient, loss-tolerant embedding techniques.
It safeguards the embedding process using a TEE and a progressive slicing mechanism.
SIDE's efficiency and effectiveness are validated through simulations and a prototype on a Creality Ender 3 printer with a Raspberry Pi 3B.

\section*{Acknowledgments}
{
    This research was supported by the National Science Foundation (Grants CNS-2223032, CNS-2038995, and CNS-223863) and the Army Research Office (Grant W911NF-24-1-0155).
}




    
    
    

\printbibliography

\begin{algorithm*}[ht]\caption{$\textsc{d-encode}$ (distinct strings encoding)}\label{alg:djEncoding}
\begin{algorithmic}[1]
\Statex \textbf{Input:}~{A binary string~$\bfw\in\bi^{l\cdot m-1}$, where~$l,m$ are positive integers such that~$m\geq \ceil{2\log l}+2$.}
\Statex \textbf{Output:}~{Array~\texttt{dStrings} of~$l$ length-$m$ pairwise distinct binary strings.}
\State Let~$\bfw\gets \bfw\circ 1$,~$i\gets0, i_\text{end}\gets l-1$\label{line:append1}, and~$\texttt{dStrings}[a]\gets\bfw[a\cdot m:(a+1)\cdot m)$ for~$a\in[0,l -1]$.\label{line:segmentW}
\While{$i < i_\text{end}$}\label{line:whileCondition}
    \State $j\gets i+1$
    \While{$j\leq i_\text{end}$}\label{line:enumerate}
        \If{$\texttt{dStrings}[i]=\texttt{dStrings}[j]$}\label{line:match}
            \State $\texttt{dStrings}\gets (\texttt{dStrings}[0:j-1],\texttt{dStrings}[ j+1:l-1])$\label{line:deleteMatched}
            \State $j'\gets 0$
            \Repeat{$j$}
                \Do
                    \State $j'\gets j'+1$\label{line:increaseBy1}
                \doWhile{$\exists r\in[l-2]~\mbox{s.t. } \texttt{dStrings}[r-1][0:\ceil{\log l}]=\textsc{binary}(j',\ceil{\log l}+1)$}\label{line:coincide}
            \EndRepeat
            \State $ \texttt{dStrings}.\textsc{append}(\textsc{binary}(j',\ceil{\log l}+1)\circ\textsc{binary}(i,\ceil{\log l})\circ 0^{m-2\ceil{\log l}-1})$\label{line:newString}
            \State $i_\text{end}\gets i_\text{end}-1$\label{line:iendminus1}
        \Else
            \State $j\gets j+1$
        \EndIf
    \EndWhile    
    \State $i\gets i+1$    
\EndWhile
\State \textbf{return}~$\texttt{dStrings}$
\end{algorithmic}
\end{algorithm*}

\appendix

\section{Distinct Strings}\label{appendix:distinct-string}

\subsection{Encoding}

Alg.~\ref{alg:djEncoding} offers a procedure that maps an input word $\bfw\in\bi^{l\cdot m -1}$ to~\texttt{dStrings}, which is an array of~$l$ pairwise distinct binary strings of length~$m$.
Besides, the inverse operation, as given in Alg.~\ref{alg:djDecoding}, outputs~$\bfw$ given the array~\texttt{dStrings}.

The encoding algorithm first appends a~$1$ to~$\bfw$ (line~\ref{line:append1}), and as a result,~$\vert\bfw\vert=l\cdot m$.
The word~$\bfw$ is segmented to~$l$ intervals of length~$m$ and placed in a tentative array~\texttt{dStrings} (line~\ref{line:segmentW}).
The segments are not necessarily pairwise distinct at this point, and they are referred as the~\emph{original}.

The encoding algorithm continues to look for identical original segments, deletes one of them, and appended a~\emph{new} binary string to~\texttt{dStrings}; this new string contains the indexing information of the two identical binary strings which allows for decoding at a later point in time.
Moreover, we end every new string with a~$0$ to distinguish it from the rightmost original segment which ends with a~$1$.

With indices~$i$ and~$i_\text{end}$ initially set to~$0$ and~$l-1$, the encoder enumerates every index~$j\in[i+1,i_\text{end}]$ for a match, i.e.,~$\texttt{dStrings}[i]$ is identical to~$\texttt{dStrings}[j]$  (line~\ref{line:enumerate}).
Once a match is found, the latter is deleted from~\texttt{dStrings} (line~\ref{line:deleteMatched}) and the indices~$i,j$ are recorded in a binary string to be placed at the end of~\texttt{dStrings}, which will be used for recovering the deleted entry during decoding.

Note that, na\"ively defining the new strings as the concatenation of binary representations of~$i$ and~$j$ may introduce more repeated strings; it may coincide with existing elements in~\texttt{dStrings}.
To this end, the algorithm looks for an~\emph{alternative} binary representation of~$j$ which is not identical to the first~$\ceil{\log l}+1$ bits of~\emph{every} existing element in~\texttt{dStrings}.

Starting from~$j'=0$, the following procedure is repeated~$j$ times.
In each time,~$j'$ is increased by~$1$ (line~\ref{line:increaseBy1}).
Then, it continues to increase until its binary representation in~$\ceil{\log l}+1$ bits does not coincide with the first $\ceil{\log l}+1$ bits of every existing element in~\texttt{dStrings} (line~\ref{line:coincide}).
One may imagine this process as looking for the~$j$-th available parking slot in a row, in which some have been occupied (unavailable).
A slot is unavailable if the binary representation of its index coincides with the first $\ceil{\log l}+1$ bits of any element in~\texttt{dStrings}.
Otherwise, it is available.
Starting for index~$0$,~$j'$ is indeed the index of the~$j$-th available slot.

Note that when the repetition stops,~$j'$ equals to the sum of~$j$ and the number of times the condition in line~\ref{line:coincide} was true.
Recall that~$1\leq j\leq l-1$, and the latter equals to the number of unavailable slots that may be encountered during the increment of~$j'$, which is at most~$l-1$ since there are only~$l-1$ elements in~\texttt{dStrings}.
Therefore,~$j'\leq 2l-2$ and can be represented by~$\ceil{\log l}+1$ bits.

\begin{algorithm*}\caption{\textsc{d-decode} (distinct strings decoding)}\label{alg:djDecoding}
\begin{algorithmic}[1]
\Statex \textbf{Input:}~{Array~\texttt{dStrings} of~$l$ length-$m$ pairwise distinct binary strings.}
\Statex \textbf{Output:}~{The information word~$\bfw\in\bi^k$ used to generate~\texttt{dStrings}}.
\While{$\texttt{dStrings}[-1][-1]=0$}\label{line:decodeWhile}
    \State $i\gets \textsc{integer}(\texttt{dStrings}[l][\ceil{\log l}+2:2\ceil{\log l}+1])$\label{line:readi}
    \State $j,j'\gets \textsc{integer}(\texttt{dStrings}[l-1][1:\ceil{\log l}+1])$\label{line:readjp}
    \ForAll{$r\in[0,l-2]$ and~$s\in[j'-1]$}\label{line:getj1}
        \If{$ \texttt{dStrings}[r][1:\ceil{\log l}+1]= \textsc{binary}(s,\ceil{\log l}+1$)} 
             $j\gets j-1$\label{line:getj2}
        \EndIf
    \EndFor
    \State $
                \texttt{dStrings}\gets (\texttt{dStrings}[0:j-1],\texttt{dStrings}[i],\texttt{dStrings}[j:l-2])$\label{line:recovery}
\EndWhile
\State $\bfw\gets\texttt{dStrings}[0]\circ\ldots\circ\texttt{dStrings}[l-1]$
\State \textbf{return}~$\bfw[:-1]$\label{line:removeLastBit}
\end{algorithmic}
\end{algorithm*}

Therefore, we use the binary representation of~$j'$ in~$\ceil{\log l}+1$ to serve as the alternative representation of~$j$.
It is concatenated with the binary representation of~$i$ (in~$\ceil{\log l}$ bits since~$i<l$) and sufficiently many~$0$'s to make a~\emph{new} string (line~\ref{line:newString}).
The new string is appended to~\texttt{dStrings}, and is different from every other element in the first~$\ceil{\log l}+1$ bits; this fact gives the following lemma.
\begin{lemma}\label{lemma:differetNewString}
    The new binary string being appended in line~\ref{line:newString} is different from every existing elements in~\texttt{dStrings}.
\end{lemma}
Lemma~\ref{lemma:differetNewString} allows us to decrease~$i_\text{end}$ by one (line~\ref{line:iendminus1}) since there is no need to compare~$\texttt{dStrings}[j]$ with element whose index is greater than~$i_\text{end}-1$.
The algorithm terminates when there are no elements of~\texttt{dStrings} remains to be compared (line~\ref{line:whileCondition}), and its output satisfies the following.

\begin{theorem}\label{theorem:pairwise-distinct}
    Algorithm~\ref{alg:djEncoding} outputs an array~\texttt{dStrings} of~$l$ pairwise distinct binary strings of length $m$.
\end{theorem}
\begin{proof}
    Assume, for sake of contradiction, that there exist~$a,b\in[0,l-1]$ and~$a<b$ such that~$\texttt{dStrings}[a]=\texttt{dStrings}[b]$.
    There are two possible cases for~$\texttt{dStrings}[b]$.
    
    If~$\texttt{dStrings}[b]$ is a new string constructed in line~\ref{line:newString}, then it is distinct from every other elements on its left by Lemma~\ref{lemma:differetNewString}, a contradiction.
    If~$\texttt{dStrings}[b]$ is not a new string, then the~$\texttt{dStrings}[a]$ on its left is not as well.
    As such,~$\texttt{dStrings}[b]$ would have been deleted in line~\ref{line:deleteMatched} when~$i=a$ and~$j=b$, a contradiction.
\end{proof}

\subsection{Decoding}

We proceed to introduce the decoding procedure in Algorithm~\ref{alg:djDecoding}, which is essentially the inverse operation of the encoding process.
Given the array~\texttt{dStrings}, the decoding algorithm reads~$i,j'$ from the rightmost element if it is a constructed new string, i.e., if last bit of which is~$0$ (line~\ref{line:decodeWhile}).
Recall that a new string is created when two identical strings is found, with~$i$ being the index of the first one (the reference) and~$j'$ being the~\emph{alternative} index of the second (the referent).

Line~\ref{line:readi} reads the value of~$i$, and line~\ref{line:readjp} reads the value of~$j'$.
Recall that~$j'$ is the index of the~$j$-th available slots in a row.
Hence, the variable~$j$ is initially set to~$j'$, and then subtracted by the number of unavailable slots with indices less than~$j'$ to reaches the actual index of the referent (line~\ref{line:getj1}--\ref{line:getj2}).
Together,~$i,j$ enable the recovery of the referent, and the rightmost element is deleted (line~\ref{line:recovery}).
Once all new strings have been consumed in the aforementioned process, the decoding is concluded and~$\bfw$ is returned (line~\ref{line:removeLastBit}).

\section{Proof of Theorem~\ref{theorem:brc}}\label{appendix:proof-of-theorem}

The crux of proving Theorem~\ref{theorem:brc} is showing that the decoder is able to obtain the key-value store~\texttt{next} from unordered and partially missing fragments which result from breaking a codeword~$\bfc$ at~$t$ arbitrary places and hiding fragments whose total length is at most~$s$ bits.
That is, the Reed-Solomon decoding in line~\ref{line:decodeRS} of Algorithm~\ref{alg:Decode} concludes successfully---it is well known that this requires the number of erasures plus twice the number of errors to be at most~$4\alpha$ in the RS codeword
\begin{equation}\label{eq:rs-cw}
(\texttt{approxNext},\bfr_1,\ldots,\bfr_{4\alpha}).
\end{equation}

Recall that every~$\alpha$-BRC codeword~$\bfc$ is a concatenation of~$\cC_\text{MU}$ codewords, which start with~$\ceil{\log m}+1$ consecutive zeros (see Section~\ref{section:preliminaries}), and redundancy packets, which are free of zero runs of length~$\ceil{\log m}+1$ or more thanks to the~$\textsc{RLL}$ encoding.
Hence, every discernible~$\cC_\text{MU}$ codeword in~$\bfc$ does not overlap with redundancy packets, and the decoder does not confuse the two.

The following lemma counts the number of erasures in~\eqref{eq:rs-cw}, which equals to the number of redundancy strings that the decoder fails to obtain from the fragments.

\begin{lemma}\label{lemma:boundErasures}
    Let~$t_2$ be the number of breaks that fall in the redundancy region, or separate the redundancy regions from the information region, and let~$s_2$ be the number of missing bits that originally reside in the redundancy region.
    Then, the number of redundancy strings that the decoder fails to obtain is at most~$4[t_2+{s_2}/{(5m+\ceil{\log m}+15})]$.
\end{lemma}
\begin{proof}
The decoding algorithm may fail to obtain a redundancy packet due to exactly one of the following reasons.
\begin{enumerate}
    \item There exists a break either in the packet itself, or in its preceding marker.
    \item The packet, as well as its preceding marker, wholly resides in a missing fragment.
\end{enumerate}

Since~$t_2$ breaks occur in the redundancy region, there are at most~$t_2$ missing redundant packets due to the first reason.

Recall that a redundancy packet and a marker add up to~$5m+\ceil{\log m}+15$ bits.
Then, for a missing fragment~$\bff$ that resides in the redundancy region (if~$\bff$ cross both regions, we only consider the part in the redundancy region), at most
\small\begin{equation*}
   {\mid\bff\mid}/{(5m+\ceil{\log m}+15)}
\end{equation*}\normalsize
packets, together with their preceding markers, are lost due to the second reason.
Together, the number of missing redundant packets is no more than
\small\begin{align*}
    &\phantom{=}t_2+\sum_{\mbox{$\bff$ in redundancy region}}\frac{\mid\bff\mid}{5m+\ceil{\log m}+15}\\
    &\leq t_2+{s_2}/{(5m+\ceil{\log m}+15)}.
\end{align*}\normalsize
Finally, since one redundancy packet contains four redundancy strings, it follows that the decoder fails to obtain at most~$4(t_2+{s_2}/{(5m+\ceil{\log m}+15)})$ redundancy strings.
\end{proof}

We continue to count the number of errors in~$\eqref{eq:rs-cw}$, i.e., the number of entries in which~\texttt{approxNext} and~\texttt{next} differ.

\begin{lemma}\label{lemma:boundErrors}
Let~$t_1$ be the number of breaks that occur in the information region, and~$s_1$ be the number of missing bits that originally reside in the information region.
Then, the number of entries in which \texttt{approxNext} and~\texttt{next} differ is at most
\begin{equation}
    2\cdot t_1 + {s_1}/{(m+\ceil{\log m}+4)}.
\end{equation}
\end{lemma}
\begin{proof}
    Observe that the decoding algorithm may fail to find a non-marker (i.e., $\bfu_i$ for~$i\geq t$) in the information region due to exactly one of the following reasons.
    \begin{enumerate}
        \item There exists a break in it.
        \item It wholly resides in a missing fragment. 
    \end{enumerate}
    
    Notice that, failing to capture a non-marker~$\bfu_u$ affects at most two key-value pairs of~\texttt{next}, i.e.,~$\texttt{next}[\bfu_{u-1}]$ and~$\texttt{next}[\bfu_{u+1}]$, and hence~$\texttt{next}$ and~$\texttt{approxNext}$ differ in at most~$2t_1$ positions due to the first reason.

    Recall that the code length of~$\cC_\text{MU}$ is~$m+\ceil{\log m}+4$.
    Hence, for a missing fragment~$\bff$ that resides in the information region (if~$\bff$ cross both regions, we only consider the part in the information region), there are at most
    \small\begin{equation*}
        {\mid\bff\mid}/{(m+\ceil{\log m}+4)}
    \end{equation*}\normalsize
    consecutive non-markers being lost due to the second reason.
    This leads to at most
    \small\begin{equation*}
        {\mid\bff\mid}/{(m+\ceil{\log m}+4)}-1
    \end{equation*}\normalsize
    different entries in~\texttt{approxNext} in addition to the ones caused by breaks. 
    Together, the number of entries that~\texttt{approxNext} and~\texttt{next} differ is at most
    \small\begin{align*}
        &2 t_1 + \sum_{\mbox{$\bff~$ in information region}}{\frac{\mid \bff\mid}{m+\ceil{\log m} +4}-1}\\
        \leq& 2 t_1 + {s_1}/{(m+\ceil{\log m} +4)}.\qedhere
    \end{align*}\normalsize
\end{proof}
Now, by Lemma~\ref{lemma:boundErasures} and Lemma~\ref{lemma:boundErrors}, the sum of the number of errors and twice the number of erasures in~\eqref{eq:rs-cw} is no more than
\small
\begin{align*}
    &\phantom{=}4(t_1+t_2)+{{4s_2}/{(5m+\ceil{\log m}+15})}+{{2s_1}/{(m+\ceil{\log m} +4})}\\
    &\leq 4t +{2s_2}{(m+\ceil{\log m} +4)}+{2s_1}/{(m+\ceil{\log m} +4)}\\
    &\leq 4 t + {2s}/{(m+\ceil{\log m} +4)} \leq 4 \alpha,
\end{align*}
\normalsize
where the last inequality follows from \eqref{eq:st-condition}.
The proof is concluded since line~\ref{line:decodeRS} outputs the correct key-value store~$\texttt{next}$, given that a~$(2^m + 4\alpha, 2^m)$ Reed-Solomon code can simultaneously correct~$x$ errors and~$y$ erasures provided that~$2x + y\le 4\alpha$.


\section{Simulation Results}\label{appendix:simulation}
This section presents simulation results on the success rate of fingerprint recovery using $\alpha$-BRC across three models.
\begin{table*}[ht]
\centering
\resizebox{.95\textwidth}{!}{\begin{tabular}{|ccccccccccccccccc|}
\hline
\multicolumn{1}{|c|}{}        & \multicolumn{4}{c|}{1-BRC, 0.613 mm/bit}                      & \multicolumn{4}{c|}{2-BRC, 0.485 mm/bit}                      & \multicolumn{4}{c|}{3-BRC, 0.401 mm/bit}                       & \multicolumn{4}{c|}{4-BRC, 0.342 mm/bit}  \\ \hline
\multicolumn{1}{|c|}{\diagbox[height=1.6em]{$\beta$}{$\rho$}} & 0\%      & 25\%     & 50\%     & \multicolumn{1}{c|}{75\%}    & 0\%      & 25\%     & 50\%     & \multicolumn{1}{c|}{75\%}    & 0\%      & 25\%     & 50\%     & \multicolumn{1}{c|}{75\%}     & 0\%      & 25\%     & 50\%     & 75\%     \\ \hline
\multicolumn{1}{|c|}{10}      & 81.25\%  & 45.41\%  & 13.87\%  & \multicolumn{1}{c|}{0.78\%}  & 100.00\% & 83.15\%  & 43.24\%  & \multicolumn{1}{c|}{6.67\%}  & 100.00\% & 97.27\%  & 76.71\%  & \multicolumn{1}{c|}{21.07\%}  & 100.00\% & 99.54\%  & 89.67\%  & 37.26\%  \\
\multicolumn{1}{|c|}{20}      & 60.16\%  & 25.98\%  & 5.10\%   & \multicolumn{1}{c|}{0.17\%}  & 99.22\%  & 85.82\%  & 46.24\%  & \multicolumn{1}{c|}{4.08\%}  & 100.00\% & 99.71\%  & 89.36\%  & \multicolumn{1}{c|}{24.80\%}  & 100.00\% & 99.90\%  & 97.05\%  & 48.36\%  \\
\multicolumn{1}{|c|}{30}      & 32.03\%  & 11.06\%  & 1.27\%   & \multicolumn{1}{c|}{0.00\%}  & 100.00\% & 85.35\%  & 40.38\%  & \multicolumn{1}{c|}{2.66\%}  & 100.00\% & 99.68\%  & 90.11\%  & \multicolumn{1}{c|}{24.41\%}  & 100.00\% & 100.00\% & 98.56\%  & 55.25\%  \\
\multicolumn{1}{|c|}{40}      & 10.94\%  & 3.34\%   & 0.56\%   & \multicolumn{1}{c|}{0.00\%}  & 96.09\%  & 75.46\%  & 32.71\%  & \multicolumn{1}{c|}{1.64\%}  & 100.00\% & 99.49\%  & 88.13\%  & \multicolumn{1}{c|}{23.97\%}  & 100.00\% & 99.98\%  & 99.00\%  & 57.52\%  \\
\multicolumn{1}{|c|}{50}      & 5.47\%   & 1.17\%   & 0.07\%   & \multicolumn{1}{c|}{0.00\%}  & 86.72\%  & 61.62\%  & 20.43\%  & \multicolumn{1}{c|}{0.32\%}  & 100.00\% & 99.58\%  & 88.11\%  & \multicolumn{1}{c|}{18.99\%}  & 100.00\% & 100.00\% & 99.29\%  & 58.30\%  \\
\multicolumn{1}{|c|}{60}      & 0.00\%   & 0.00\%   & 0.00\%   & \multicolumn{1}{c|}{0.00\%}  & 78.91\%  & 43.09\%  & 10.52\%  & \multicolumn{1}{c|}{0.27\%}  & 100.00\% & 99.44\%  & 83.54\%  & \multicolumn{1}{c|}{14.77\%}  & 100.00\% & 99.98\%  & 98.80\%  & 53.00\%  \\
\multicolumn{1}{|c|}{70}      & 0.00\%   & 0.00\%   & 0.00\%   & \multicolumn{1}{c|}{0.00\%}  & 54.69\%  & 26.05\%  & 4.83\%   & \multicolumn{1}{c|}{0.20\%}  & 100.00\% & 98.90\%  & 79.42\%  & \multicolumn{1}{c|}{12.13\%}  & 100.00\% & 100.00\% & 99.02\%  & 50.34\%  \\
\multicolumn{1}{|c|}{80}      & 0.00\%   & 0.00\%   & 0.00\%   & \multicolumn{1}{c|}{0.00\%}  & 50.78\%  & 20.95\%  & 3.56\%   & \multicolumn{1}{c|}{0.05\%}  & 100.00\% & 97.68\%  & 71.78\%  & \multicolumn{1}{c|}{8.11\%}   & 100.00\% & 99.98\%  & 98.17\%  & 46.80\%  \\
\multicolumn{1}{|c|}{90}      & 0.00\%   & 0.00\%   & 0.00\%   & \multicolumn{1}{c|}{0.00\%}  & 37.50\%  & 13.04\%  & 2.17\%   & \multicolumn{1}{c|}{0.02\%}  & 100.00\% & 95.97\%  & 64.60\%  & \multicolumn{1}{c|}{5.27\%}   & 100.00\% & 99.98\%  & 98.24\%  & 43.12\%  \\
\multicolumn{1}{|c|}{100}     & 0.00\%   & 0.00\%   & 0.00\%   & \multicolumn{1}{c|}{0.00\%}  & 14.06\%  & 4.25\%   & 0.29\%   & \multicolumn{1}{c|}{0.00\%}  & 100.00\% & 90.70\%  & 53.10\%  & \multicolumn{1}{c|}{3.64\%}   & 100.00\% & 99.93\%  & 96.75\%  & 33.59\%  \\ \hline
\multicolumn{1}{|c|}{}        & \multicolumn{4}{c|}{5-BRC, 0.298 mm/bit}                      & \multicolumn{4}{c|}{6-BRC, 0.264 mm/bit}                      & \multicolumn{4}{c|}{7-BRC, 0.237 mm/bit}                       & \multicolumn{4}{c|}{8-BRC, 0.215 mm/bit}  \\ \hline
\multicolumn{1}{|c|}{\diagbox[height=1.6em]{$\beta$}{$\rho$}} & 0\%      & 25\%     & 50\%     & \multicolumn{1}{c|}{75\%}    & 0\%      & 25\%     & 50\%     & \multicolumn{1}{c|}{75\%}    & 0\%      & 25\%     & 50\%     & \multicolumn{1}{c|}{75\%}     & 0\%      & 25\%     & 50\%     & 75\%     \\ \hline
\multicolumn{1}{|c|}{10}      & 100.00\% & 99.80\%  & 96.00\%  & \multicolumn{1}{c|}{56.30\%} & 100.00\% & 100.00\% & 97.63\%  & \multicolumn{1}{c|}{66.60\%} & 100.00\% & 99.95\%  & 99.05\%  & \multicolumn{1}{c|}{74.51\%}  & 100.00\% & 100.00\% & 99.61\%  & 81.74\%  \\
\multicolumn{1}{|c|}{20}      & 100.00\% & 100.00\% & 99.51\%  & \multicolumn{1}{c|}{74.41\%} & 100.00\% & 100.00\% & 99.73\%  & \multicolumn{1}{c|}{84.30\%} & 100.00\% & 100.00\% & 99.95\%  & \multicolumn{1}{c|}{92.19\%}  & 100.00\% & 100.00\% & 100.00\% & 95.95\%  \\
\multicolumn{1}{|c|}{30}      & 100.00\% & 100.00\% & 99.93\%  & \multicolumn{1}{c|}{82.54\%} & 100.00\% & 100.00\% & 99.98\%  & \multicolumn{1}{c|}{93.16\%} & 100.00\% & 100.00\% & 100.00\% & \multicolumn{1}{c|}{97.53\%}  & 100.00\% & 100.00\% & 100.00\% & 99.17\%  \\
\multicolumn{1}{|c|}{40}      & 100.00\% & 100.00\% & 99.98\%  & \multicolumn{1}{c|}{87.52\%} & 100.00\% & 100.00\% & 100.00\% & \multicolumn{1}{c|}{96.29\%} & 100.00\% & 100.00\% & 100.00\% & \multicolumn{1}{c|}{98.58\%}  & 100.00\% & 100.00\% & 100.00\% & 99.80\%  \\
\multicolumn{1}{|c|}{50}      & 100.00\% & 100.00\% & 99.93\%  & \multicolumn{1}{c|}{88.94\%} & 100.00\% & 100.00\% & 100.00\% & \multicolumn{1}{c|}{97.53\%} & 100.00\% & 100.00\% & 100.00\% & \multicolumn{1}{c|}{99.49\%}  & 100.00\% & 100.00\% & 100.00\% & 99.95\%  \\
\multicolumn{1}{|c|}{60}      & 100.00\% & 100.00\% & 99.98\%  & \multicolumn{1}{c|}{87.99\%} & 100.00\% & 100.00\% & 100.00\% & \multicolumn{1}{c|}{98.24\%} & 100.00\% & 100.00\% & 100.00\% & \multicolumn{1}{c|}{99.71\%}  & 100.00\% & 100.00\% & 100.00\% & 99.95\%  \\
\multicolumn{1}{|c|}{70}      & 100.00\% & 100.00\% & 100.00\% & \multicolumn{1}{c|}{90.04\%} & 100.00\% & 100.00\% & 100.00\% & \multicolumn{1}{c|}{98.95\%} & 100.00\% & 100.00\% & 100.00\% & \multicolumn{1}{c|}{99.95\%}  & 100.00\% & 100.00\% & 100.00\% & 100.00\% \\
\multicolumn{1}{|c|}{80}      & 100.00\% & 100.00\% & 100.00\% & \multicolumn{1}{c|}{89.70\%} & 100.00\% & 100.00\% & 100.00\% & \multicolumn{1}{c|}{99.19\%} & 100.00\% & 100.00\% & 100.00\% & \multicolumn{1}{c|}{99.95\%}  & 100.00\% & 100.00\% & 100.00\% & 99.98\%  \\
\multicolumn{1}{|c|}{90}      & 100.00\% & 100.00\% & 99.98\%  & \multicolumn{1}{c|}{87.94\%} & 100.00\% & 100.00\% & 100.00\% & \multicolumn{1}{c|}{98.80\%} & 100.00\% & 100.00\% & 100.00\% & \multicolumn{1}{c|}{99.88\%}  & 100.00\% & 100.00\% & 100.00\% & 100.00\% \\
\multicolumn{1}{|c|}{100}     & 100.00\% & 100.00\% & 99.93\%  & \multicolumn{1}{c|}{85.28\%} & 100.00\% & 100.00\% & 100.00\% & \multicolumn{1}{c|}{99.07\%} & 100.00\% & 100.00\% & 100.00\% & \multicolumn{1}{c|}{100.00\%} & 100.00\% & 100.00\% & 100.00\% & 100.00\% \\ \hline
\end{tabular}%
}
\caption{Simulation results of the FMDA Glock Frame.}
\label{tab:simulation-glock}
\end{table*}

\begin{table*}[ht]
\centering
\resizebox{.95\textwidth}{!}{\begin{tabular}{|ccccccccccccccccc|}
\hline
\multicolumn{1}{|c|}{}        & \multicolumn{4}{c|}{1-BRC, 0.653 mm/bit}                      & \multicolumn{4}{c|}{2-BRC, 0.5168 mm/bit}                     & \multicolumn{4}{c|}{3-BRC, 0.427 mm/bit}                       & \multicolumn{4}{c|}{4-BRC, 0.364 mm/bit}  \\ \hline
\multicolumn{1}{|c|}{\diagbox[height=1.6em]{$\beta$}{$\rho$}} & 0\%      & 25\%     & 50\%     & \multicolumn{1}{c|}{75\%}    & 0\%      & 25\%     & 50\%     & \multicolumn{1}{c|}{75\%}    & 0\%      & 25\%     & 50\%     & \multicolumn{1}{c|}{75\%}     & 0\%      & 25\%     & 50\%     & 75\%     \\ \hline
\multicolumn{1}{|c|}{10}      & 80.47\%  & 42.77\%  & 10.55\%  & \multicolumn{1}{c|}{0.44\%}  & 100.00\% & 89.14\%  & 53.91\%  & \multicolumn{1}{c|}{8.94\%}  & 100.00\% & 97.19\%  & 78.37\%  & \multicolumn{1}{c|}{23.56\%}  & 100.00\% & 99.10\%  & 88.96\%  & 41.82\%  \\
\multicolumn{1}{|c|}{20}      & 67.97\%  & 31.69\%  & 6.71\%   & \multicolumn{1}{c|}{0.20\%}  & 100.00\% & 95.39\%  & 63.92\%  & \multicolumn{1}{c|}{8.52\%}  & 100.00\% & 99.73\%  & 92.99\%  & \multicolumn{1}{c|}{36.82\%}  & 100.00\% & 99.98\%  & 98.29\%  & 60.50\%  \\
\multicolumn{1}{|c|}{30}      & 44.53\%  & 16.70\%  & 2.10\%   & \multicolumn{1}{c|}{0.02\%}  & 100.00\% & 93.92\%  & 59.42\%  & \multicolumn{1}{c|}{5.30\%}  & 100.00\% & 99.95\%  & 95.92\%  & \multicolumn{1}{c|}{36.35\%}  & 100.00\% & 100.00\% & 99.44\%  & 68.36\%  \\
\multicolumn{1}{|c|}{40}      & 30.47\%  & 10.55\%  & 1.20\%   & \multicolumn{1}{c|}{0.00\%}  & 99.22\%  & 91.26\%  & 50.76\%  & \multicolumn{1}{c|}{3.71\%}  & 100.00\% & 100.00\% & 96.95\%  & \multicolumn{1}{c|}{36.38\%}  & 100.00\% & 100.00\% & 99.85\%  & 75.17\%  \\
\multicolumn{1}{|c|}{50}      & 15.62\%  & 4.88\%   & 0.44\%   & \multicolumn{1}{c|}{0.02\%}  & 98.44\%  & 85.33\%  & 40.62\%  & \multicolumn{1}{c|}{1.93\%}  & 100.00\% & 99.95\%  & 95.53\%  & \multicolumn{1}{c|}{31.62\%}  & 100.00\% & 100.00\% & 99.95\%  & 73.32\%  \\
\multicolumn{1}{|c|}{60}      & 7.81\%   & 1.49\%   & 0.10\%   & \multicolumn{1}{c|}{0.00\%}  & 96.09\%  & 78.15\%  & 30.62\%  & \multicolumn{1}{c|}{1.15\%}  & 100.00\% & 99.95\%  & 96.39\%  & \multicolumn{1}{c|}{31.05\%}  & 100.00\% & 100.00\% & 99.78\%  & 72.36\%  \\
\multicolumn{1}{|c|}{70}      & 2.34\%   & 0.51\%   & 0.00\%   & \multicolumn{1}{c|}{0.00\%}  & 95.31\%  & 69.09\%  & 21.56\%  & \multicolumn{1}{c|}{0.29\%}  & 100.00\% & 100.00\% & 94.78\%  & \multicolumn{1}{c|}{26.10\%}  & 100.00\% & 100.00\% & 99.83\%  & 70.70\%  \\
\multicolumn{1}{|c|}{80}      & 0.00\%   & 0.00\%   & 0.00\%   & \multicolumn{1}{c|}{0.00\%}  & 88.28\%  & 55.64\%  & 13.94\%  & \multicolumn{1}{c|}{0.27\%}  & 100.00\% & 99.98\%  & 91.85\%  & \multicolumn{1}{c|}{20.14\%}  & 100.00\% & 100.00\% & 99.93\%  & 67.07\%  \\
\multicolumn{1}{|c|}{90}      & 0.00\%   & 0.00\%   & 0.00\%   & \multicolumn{1}{c|}{0.00\%}  & 75.78\%  & 42.65\%  & 7.89\%   & \multicolumn{1}{c|}{0.00\%}  & 100.00\% & 99.63\%  & 88.13\%  & \multicolumn{1}{c|}{15.84\%}  & 100.00\% & 100.00\% & 99.68\%  & 63.35\%  \\
\multicolumn{1}{|c|}{100}     & 0.00\%   & 0.00\%   & 0.00\%   & \multicolumn{1}{c|}{0.00\%}  & 67.97\%  & 31.10\%  & 4.54\%   & \multicolumn{1}{c|}{0.02\%}  & 100.00\% & 99.63\%  & 84.64\%  & \multicolumn{1}{c|}{11.47\%}  & 100.00\% & 100.00\% & 99.78\%  & 61.60\%  \\ \hline
\multicolumn{1}{|c|}{}        & \multicolumn{4}{c|}{5-BRC, 0.318 mm/bit}                      & \multicolumn{4}{c|}{6-BRC, 0.281 mm/bit}                      & \multicolumn{4}{c|}{7-BRC, 0.253 mm/bit}                       & \multicolumn{4}{c|}{8-BRC, 0.229 mm/bit}  \\ \hline
\multicolumn{1}{|c|}{\diagbox[height=1.6em]{$\beta$}{$\rho$}} & 0\%      & 25\%     & 50\%     & \multicolumn{1}{c|}{75\%}    & 0\%      & 25\%     & 50\%     & \multicolumn{1}{c|}{75\%}    & 0\%      & 25\%     & 50\%     & \multicolumn{1}{c|}{75\%}     & 0\%      & 25\%     & 50\%     & 75\%     \\ \hline
\multicolumn{1}{|c|}{10}      & 100.00\% & 99.39\%  & 93.92\%  & \multicolumn{1}{c|}{53.44\%} & 100.00\% & 99.66\%  & 96.75\%  & \multicolumn{1}{c|}{62.30\%} & 100.00\% & 99.85\%  & 97.63\%  & \multicolumn{1}{c|}{70.87\%}  & 100.00\% & 99.61\%  & 98.34\%  & 76.86\%  \\
\multicolumn{1}{|c|}{20}      & 100.00\% & 99.98\%  & 99.68\%  & \multicolumn{1}{c|}{78.27\%} & 100.00\% & 100.00\% & 99.98\%  & \multicolumn{1}{c|}{87.48\%} & 100.00\% & 100.00\% & 99.98\%  & \multicolumn{1}{c|}{92.41\%}  & 100.00\% & 100.00\% & 100.00\% & 95.48\%  \\
\multicolumn{1}{|c|}{30}      & 100.00\% & 100.00\% & 99.95\%  & \multicolumn{1}{c|}{86.04\%} & 100.00\% & 100.00\% & 99.98\%  & \multicolumn{1}{c|}{94.41\%} & 100.00\% & 100.00\% & 100.00\% & \multicolumn{1}{c|}{97.49\%}  & 100.00\% & 100.00\% & 100.00\% & 98.75\%  \\
\multicolumn{1}{|c|}{40}      & 100.00\% & 100.00\% & 100.00\% & \multicolumn{1}{c|}{90.58\%} & 100.00\% & 100.00\% & 100.00\% & \multicolumn{1}{c|}{97.53\%} & 100.00\% & 100.00\% & 100.00\% & \multicolumn{1}{c|}{99.54\%}  & 100.00\% & 100.00\% & 100.00\% & 99.83\%  \\
\multicolumn{1}{|c|}{50}      & 100.00\% & 100.00\% & 100.00\% & \multicolumn{1}{c|}{91.80\%} & 100.00\% & 100.00\% & 100.00\% & \multicolumn{1}{c|}{98.12\%} & 100.00\% & 100.00\% & 100.00\% & \multicolumn{1}{c|}{99.66\%}  & 100.00\% & 100.00\% & 100.00\% & 99.88\%  \\
\multicolumn{1}{|c|}{60}      & 100.00\% & 100.00\% & 100.00\% & \multicolumn{1}{c|}{93.77\%} & 100.00\% & 100.00\% & 100.00\% & \multicolumn{1}{c|}{98.54\%} & 100.00\% & 100.00\% & 100.00\% & \multicolumn{1}{c|}{99.90\%}  & 100.00\% & 100.00\% & 100.00\% & 99.98\%  \\
\multicolumn{1}{|c|}{70}      & 100.00\% & 100.00\% & 100.00\% & \multicolumn{1}{c|}{92.77\%} & 100.00\% & 100.00\% & 100.00\% & \multicolumn{1}{c|}{98.97\%} & 100.00\% & 100.00\% & 100.00\% & \multicolumn{1}{c|}{99.80\%}  & 100.00\% & 100.00\% & 100.00\% & 99.98\%  \\
\multicolumn{1}{|c|}{80}      & 100.00\% & 100.00\% & 100.00\% & \multicolumn{1}{c|}{93.48\%} & 100.00\% & 100.00\% & 100.00\% & \multicolumn{1}{c|}{99.10\%} & 100.00\% & 100.00\% & 100.00\% & \multicolumn{1}{c|}{99.88\%}  & 100.00\% & 100.00\% & 100.00\% & 100.00\% \\
\multicolumn{1}{|c|}{90}      & 100.00\% & 100.00\% & 100.00\% & \multicolumn{1}{c|}{93.31\%} & 100.00\% & 100.00\% & 100.00\% & \multicolumn{1}{c|}{99.44\%} & 100.00\% & 100.00\% & 100.00\% & \multicolumn{1}{c|}{100.00\%} & 100.00\% & 100.00\% & 100.00\% & 100.00\% \\
\multicolumn{1}{|c|}{100}     & 100.00\% & 100.00\% & 100.00\% & \multicolumn{1}{c|}{93.09\%} & 100.00\% & 100.00\% & 100.00\% & \multicolumn{1}{c|}{99.24\%} & 100.00\% & 100.00\% & 100.00\% & \multicolumn{1}{c|}{99.98\%}  & 100.00\% & 100.00\% & 100.00\% & 100.00\% \\ \hline
\end{tabular}%
}
\caption{Simulation results of AR-15 lower receiver.}
\label{tab:simulation-ar}
\end{table*}

\begin{table*}[ht]
\centering
\resizebox{.95\textwidth}{!}{\begin{tabular}{|ccccccccccccccccc|}
\hline
\multicolumn{1}{|c|}{}    & \multicolumn{4}{c|}{1-BRC, 0.161 mm/bit}                      & \multicolumn{4}{c|}{2-BRC, 0.127 mm/bit}                      & \multicolumn{4}{c|}{3-BRC, 0.105 mm/bit}                      & \multicolumn{4}{c|}{4-BRC, 0.090 mm/bit}  \\ \hline
\multicolumn{1}{|c|}{\diagbox[height=1.6em]{$\beta$}{$\rho$}}    & 0\%      & 25\%     & 50\%     & \multicolumn{1}{c|}{75\%}    & 0\%      & 25\%     & 50\%     & \multicolumn{1}{c|}{75\%}    & 0\%      & 25\%     & 50\%     & \multicolumn{1}{c|}{75\%}    & 0\%      & 25\%     & 50\%     & 75\%     \\ \hline
\multicolumn{1}{|c|}{10}  & 100.00\% & 78.34\%  & 53.69\%  & \multicolumn{1}{c|}{20.41\%} & 100.00\% & 92.63\%  & 74.68\%  & \multicolumn{1}{c|}{39.60\%} & 100.00\% & 94.19\%  & 77.69\%  & \multicolumn{1}{c|}{44.75\%} & 100.00\% & 95.65\%  & 81.69\%  & 49.51\%  \\
\multicolumn{1}{|c|}{20}  & 100.00\% & 79.91\%  & 56.13\%  & \multicolumn{1}{c|}{19.97\%} & 100.00\% & 96.78\%  & 81.64\%  & \multicolumn{1}{c|}{43.51\%} & 100.00\% & 97.80\%  & 88.70\%  & \multicolumn{1}{c|}{57.37\%} & 100.00\% & 98.27\%  & 90.84\%  & 62.26\%  \\
\multicolumn{1}{|c|}{30}  & 100.00\% & 83.23\%  & 57.28\%  & \multicolumn{1}{c|}{18.53\%} & 100.00\% & 98.85\%  & 89.18\%  & \multicolumn{1}{c|}{48.44\%} & 100.00\% & 99.29\%  & 93.87\%  & \multicolumn{1}{c|}{66.09\%} & 100.00\% & 99.76\%  & 95.70\%  & 73.36\%  \\
\multicolumn{1}{|c|}{40}  & 100.00\% & 84.23\%  & 57.32\%  & \multicolumn{1}{c|}{19.60\%} & 100.00\% & 99.41\%  & 91.43\%  & \multicolumn{1}{c|}{49.66\%} & 100.00\% & 99.83\%  & 95.95\%  & \multicolumn{1}{c|}{70.80\%} & 100.00\% & 99.93\%  & 97.53\%  & 79.71\%  \\
\multicolumn{1}{|c|}{50}  & 99.22\%  & 84.03\%  & 56.35\%  & \multicolumn{1}{c|}{16.21\%} & 100.00\% & 99.80\%  & 92.75\%  & \multicolumn{1}{c|}{50.88\%} & 100.00\% & 99.98\%  & 97.78\%  & \multicolumn{1}{c|}{74.83\%} & 100.00\% & 99.93\%  & 98.68\%  & 83.67\%  \\
\multicolumn{1}{|c|}{60}  & 98.44\%  & 83.98\%  & 53.76\%  & \multicolumn{1}{c|}{14.60\%} & 100.00\% & 99.76\%  & 93.02\%  & \multicolumn{1}{c|}{51.44\%} & 100.00\% & 99.93\%  & 98.41\%  & \multicolumn{1}{c|}{77.10\%} & 100.00\% & 99.98\%  & 99.19\%  & 86.33\%  \\
\multicolumn{1}{|c|}{70}  & 96.09\%  & 79.42\%  & 48.00\%  & \multicolumn{1}{c|}{11.52\%} & 100.00\% & 99.83\%  & 93.41\%  & \multicolumn{1}{c|}{49.95\%} & 100.00\% & 99.98\%  & 98.90\%  & \multicolumn{1}{c|}{79.66\%} & 100.00\% & 100.00\% & 99.61\%  & 88.92\%  \\
\multicolumn{1}{|c|}{80}  & 96.09\%  & 78.74\%  & 50.46\%  & \multicolumn{1}{c|}{10.57\%} & 100.00\% & 99.63\%  & 92.85\%  & \multicolumn{1}{c|}{47.73\%} & 100.00\% & 100.00\% & 98.97\%  & \multicolumn{1}{c|}{79.17\%} & 100.00\% & 100.00\% & 99.83\%  & 91.33\%  \\
\multicolumn{1}{|c|}{90}  & 92.97\%  & 76.07\%  & 43.09\%  & \multicolumn{1}{c|}{8.79\%}  & 100.00\% & 99.58\%  & 92.53\%  & \multicolumn{1}{c|}{46.17\%} & 100.00\% & 100.00\% & 99.39\%  & \multicolumn{1}{c|}{81.54\%} & 100.00\% & 100.00\% & 99.90\%  & 92.72\%  \\
\multicolumn{1}{|c|}{100} & 91.41\%  & 70.19\%  & 37.35\%  & \multicolumn{1}{c|}{6.54\%}  & 100.00\% & 99.41\%  & 91.14\%  & \multicolumn{1}{c|}{43.51\%} & 100.00\% & 100.00\% & 99.32\%  & \multicolumn{1}{c|}{80.52\%} & 100.00\% & 100.00\% & 99.95\%  & 92.65\%  \\ \hline
\multicolumn{1}{|c|}{}    & \multicolumn{4}{c|}{5-BRC, 0.078 mm/bit}                      & \multicolumn{4}{c|}{6-BRC, 0.062 mm/bit}                      & \multicolumn{4}{c|}{7-BRC, 0.062 mm/bit}                      & \multicolumn{4}{c|}{8-BRC, 0.056 mm/bit}  \\ \hline
\multicolumn{1}{|c|}{\diagbox[height=1.6em]{$\beta$}{$\rho$}}    & 0\%      & 25\%     & 50\%     & \multicolumn{1}{c|}{75\%}    & 0\%      & 25\%     & 50\%     & \multicolumn{1}{c|}{75\%}    & 0\%      & 25\%     & 50\%     & \multicolumn{1}{c|}{75\%}    & 0\%      & 25\%     & 50\%     & 75\%     \\ \hline
\multicolumn{1}{|c|}{10}  & 100.00\% & 97.31\%  & 87.40\%  & \multicolumn{1}{c|}{58.37\%} & 100.00\% & 99.73\%  & 95.51\%  & \multicolumn{1}{c|}{71.73\%} & 100.00\% & 99.98\%  & 98.22\%  & \multicolumn{1}{c|}{82.42\%} & 100.00\% & 99.98\%  & 99.46\%  & 89.50\%  \\
\multicolumn{1}{|c|}{20}  & 100.00\% & 99.49\%  & 95.39\%  & \multicolumn{1}{c|}{72.34\%} & 100.00\% & 100.00\% & 99.15\%  & \multicolumn{1}{c|}{86.82\%} & 100.00\% & 100.00\% & 99.90\%  & \multicolumn{1}{c|}{93.21\%} & 100.00\% & 100.00\% & 99.98\%  & 96.36\%  \\
\multicolumn{1}{|c|}{30}  & 100.00\% & 99.80\%  & 97.88\%  & \multicolumn{1}{c|}{81.01\%} & 100.00\% & 99.98\%  & 99.61\%  & \multicolumn{1}{c|}{90.84\%} & 100.00\% & 100.00\% & 100.00\% & \multicolumn{1}{c|}{97.05\%} & 100.00\% & 100.00\% & 100.00\% & 98.44\%  \\
\multicolumn{1}{|c|}{40}  & 100.00\% & 99.95\%  & 98.97\%  & \multicolumn{1}{c|}{86.82\%} & 100.00\% & 100.00\% & 99.80\%  & \multicolumn{1}{c|}{94.87\%} & 100.00\% & 100.00\% & 99.98\%  & \multicolumn{1}{c|}{98.56\%} & 100.00\% & 100.00\% & 100.00\% & 99.46\%  \\
\multicolumn{1}{|c|}{50}  & 100.00\% & 100.00\% & 99.37\%  & \multicolumn{1}{c|}{90.80\%} & 100.00\% & 100.00\% & 99.88\%  & \multicolumn{1}{c|}{96.02\%} & 100.00\% & 100.00\% & 99.98\%  & \multicolumn{1}{c|}{99.10\%} & 100.00\% & 100.00\% & 100.00\% & 99.78\%  \\
\multicolumn{1}{|c|}{60}  & 100.00\% & 100.00\% & 99.73\%  & \multicolumn{1}{c|}{93.92\%} & 100.00\% & 100.00\% & 99.95\%  & \multicolumn{1}{c|}{97.73\%} & 100.00\% & 100.00\% & 100.00\% & \multicolumn{1}{c|}{99.41\%} & 100.00\% & 100.00\% & 100.00\% & 99.90\%  \\
\multicolumn{1}{|c|}{70}  & 100.00\% & 100.00\% & 99.83\%  & \multicolumn{1}{c|}{95.61\%} & 100.00\% & 100.00\% & 99.98\%  & \multicolumn{1}{c|}{98.39\%} & 100.00\% & 100.00\% & 100.00\% & \multicolumn{1}{c|}{99.46\%} & 100.00\% & 100.00\% & 100.00\% & 99.98\%  \\
\multicolumn{1}{|c|}{80}  & 100.00\% & 100.00\% & 99.95\%  & \multicolumn{1}{c|}{95.95\%} & 100.00\% & 100.00\% & 100.00\% & \multicolumn{1}{c|}{98.97\%} & 100.00\% & 100.00\% & 100.00\% & \multicolumn{1}{c|}{99.78\%} & 100.00\% & 100.00\% & 100.00\% & 99.93\%  \\
\multicolumn{1}{|c|}{90}  & 100.00\% & 100.00\% & 100.00\% & \multicolumn{1}{c|}{97.41\%} & 100.00\% & 100.00\% & 100.00\% & \multicolumn{1}{c|}{99.19\%} & 100.00\% & 100.00\% & 100.00\% & \multicolumn{1}{c|}{99.88\%} & 100.00\% & 100.00\% & 100.00\% & 100.00\% \\
\multicolumn{1}{|c|}{100} & 100.00\% & 100.00\% & 99.95\%  & \multicolumn{1}{c|}{97.63\%} & 100.00\% & 100.00\% & 100.00\% & \multicolumn{1}{c|}{99.44\%} & 100.00\% & 100.00\% & 100.00\% & \multicolumn{1}{c|}{99.90\%} & 100.00\% & 100.00\% & 100.00\% & 100.00\% \\ \hline
\end{tabular}%
}
\caption{Simulation results of 3DBenchy.}
\label{tab:simulation-benchy}
\end{table*}

\end{document}